\documentclass[letterpaper,11pt]{article}

\usepackage{dsfont}
\usepackage{appendix}
\usepackage{lipsum}
\usepackage[colorlinks]{hyperref}
\usepackage[auth-lg]{authblk}
\usepackage{tcolorbox}
\usepackage{kvoptions}
\usepackage{xparse}
\usepackage{color}
\usepackage{etoolbox}
\usepackage{paralist}
\usepackage[framemethod=TikZ, xcolor={RGB,table}]{mdframed}
\usepackage{mathtools}
\usepackage{todonotes}
\usepackage{xcolor}  
\usepackage{url}
\usepackage{amsmath}

\usepackage{fancyhdr}
\usepackage{colonequals}
\usepackage{yfonts}
\usepackage{graphicx}
\usepackage{amstext,amssymb,amsmath,amsthm}
\usepackage[nice]{nicefrac}

\usepackage{mathrsfs}
\usepackage{caption}
\usepackage{centernot}
\usepackage{tikz}
\usetikzlibrary{arrows}
\usepackage{microtype}
\usepackage{algorithm,algpseudocode}
\usepackage[flushmargin,norule]{footmisc}
\usepackage[letterpaper,margin=1in]{geometry}
\usepackage{dsfont}
\usepackage{tcolorbox}

\usetikzlibrary{automata,positioning}

\usepackage{thmtools}
\usepackage{thm-restate}
\usepackage{cleveref}
\usepackage{enumitem}
 \usepackage{optidef}

\tcbuselibrary{theorems}

\hypersetup{linkcolor=blue,citecolor=blue,urlcolor=blue}

\tcbuselibrary{listings,breakable}
\allowdisplaybreaks

\newcommand{\last}[2]{\textbf{last}^{#1}(#2)}

\DeclareMathOperator*{\Exp}{\ensuremath{{\mathbb{E}}}}
\DeclareMathOperator*{\Prob}{\ensuremath{\textnormal{Pr}}}
\renewcommand{\Pr}{\Prob}

\newcommand{\prob}[1]{\Pr\bracket{{#1}}}
\newcommand{\card}[1]{\left|#1\right|}
\newcommand{\paren}[1]{\left ( #1 \right )}
\newcommand{\bracket}[1]{\left [ #1 \right ]}
\newcommand{\expect}[1]{\Exp\bracket{#1}}

\newcommand{\set}[1]{\ensuremath{\left\{ #1 \right\}}}

\newcommand{\eps}{\varepsilon}

\theoremstyle{definition}
\newtheorem{definition2}{Definition}
\theoremstyle{definition}
\newtheorem{example2}[definition2]{Example}
\theoremstyle{definition}

\theoremstyle{definition}
\newtheorem{theorem2}[definition2]{Theorem}
\theoremstyle{definition}
\newtheorem{lemma2}[definition2]{Lemma}
\theoremstyle{definition}
\newtheorem{corollary2}[definition2]{Corollary}
\theoremstyle{definition}
\newtheorem{observation2}[definition2]{Observation}
\theoremstyle{definition}
\newtheorem{claim2}[definition2]{Claim}
\theoremstyle{definition}
\newtheorem{note2}[definition2]{Note}
\theoremstyle{definition}
\newtheorem{remark2}[definition2]{Remark}
\theoremstyle{definition}
\newtheorem{research2}[definition2]{Research Direction}
\theoremstyle{definition}
\newtheorem{conjecture2}[definition2]{Conjecture}
\theoremstyle{definition}
\newtheorem{proposition}[definition2]{Proposition}
\theoremstyle{definition}
\usepackage{enumitem}
\setenumerate[0]{label=(\alph*)} 

\newenvironment{tbox}{\begin{tcolorbox}[
		enlarge top by=3pt,
		enlarge bottom by=3pt,
		boxsep=0pt,
		left=4pt,
		right=4pt,
		top=10pt,
		arc=0pt,
		boxrule=1pt,toprule=1pt,
		colback=blue!2,
		colframe=blue,
		]
	}
{\end{tcolorbox}}

\newenvironment{tbox2}{\begin{tcolorbox}[
		enlarge top by=3pt,
		enlarge bottom by=3pt,
		breakable,
		boxsep=0pt,
		left=4pt,
		right=4pt,
		top=10pt,
		arc=0pt,
		boxrule=1pt,toprule=1pt,
		colback=blue!2,
		colframe=blue,
		]
	}
	{\end{tcolorbox}}

\newtheorem{mdalg}{Algorithm}

\newcommand{\erase}[1]{}

\newcommand{\gamestate}{\mathcal{S}}

\newcommand{\timesteps}{T}
\newcommand{\game}{\Phi}
\newcommand{\final}{m}
\newcommand{\phirandom}{\phi^R}
\newcommand{\phidense}{\phi^D}

\newcommand{\obl}{{\bf obl}}
\definecolor{celadon}{rgb}{0.67, 0.88, 0.69}

\title{Randomized Greedy Online Edge Coloring Succeeds for Dense and Randomly-Ordered Graphs}

\author[1]{Aditi Dudeja\thanks{aditi.dudeja@plus.ac.at. Supported by Austrian
Science Fund (FWF): P 32863-N and ERC grant agreement 947702.}}
\affil{University of Salzburg}
\author[2]{Rashmika Goswami \thanks{rg894@scarletmail.rutgers.edu.}}
\author[2]{Michael Saks \thanks{saks@math.rutgers.edu.}}
\affil{Rutgers University}

\date{}
\begin{document}
\begin{titlepage}
\maketitle
\begin{abstract}
     Vizing's theorem states that any graph of maximum degree $\Delta$ can be properly edge colored with at most $\Delta+1$ colors. In the online setting, it has been a matter of interest to find an algorithm that can properly edge color any graph on $n$ vertices with maximum degree $\Delta = \omega(\log n)$ using at most $(1+o(1))\Delta$ colors. Here we study the na\"{i}ve random greedy algorithm, which simply chooses a legal color uniformly at random for each edge upon arrival. We show that this algorithm can $(1+\epsilon)\Delta$-color the graph for arbitrary $\epsilon$ in two contexts: first, if the edges arrive in a uniformly random order, and second, if the edges arrive in an adversarial order but the graph is sufficiently dense, i.e., $n = O(\Delta)$. Prior to this work, the random greedy algorithm was only known to succeed in trees.
     
     Our second result is applicable even when the adversary is \emph{adaptive}, and therefore implies the existence of a deterministic edge coloring algorithm which $(1+\epsilon)\Delta$ edge colors a dense graph. Prior to this, the best known deterministic algorithm for this problem was the simple greedy algorithm which utilized $2\Delta-1$ colors.
\end{abstract}
 \thispagestyle{empty}
\newpage
\end{titlepage}

\tableofcontents
\newpage

\section{Introduction}

The edge coloring problems for graphs is to assign colors to the edges of a given graph so that any two edges meeting at a vertex are assigned different colors. Trivially, the number of colors needed is at least the maximum vertex degree $\Delta$. A theorem of Vizing's states that every graph can be properly edge colored using $\Delta+1$ colors. Vizing's proof is constructive, and gives an algorithm to $\Delta+1$ color a graph with $n$ vertices and $m$ edges in $O(mn)$ time (see \cite{MisraG92}). 

Our focus is on the edge coloring problem in the online setting, which was first introduced by Bar-Noy, Motwani and Naor \cite{Bar-NoyMN92}. Starting from the empty graph on vertex set $V$, edges of the graph are revealed one at a time, and the algorithm must irrevocably assign colors to the edges as they arrive. The simplest online edge coloring algorithms are greedy algorithms, which color each arriving edge by some previously used color whenever it is possible to do so. Since every arriving edge touches at most $2\Delta-2$ existing edges, any greedy algorithm uses at most $2\Delta-1$ colors. In their paper, \cite{Bar-NoyMN92} showed that for $\Delta=O(\log n)$, no deterministic online algorithm can guarantee better than $2\Delta-1$ coloring and thus greedy algorithms are optimal in this case. They also showed that for $\Delta=O(\sqrt{\log n})$, no randomized online algorithm can achieve a better than $2\Delta-1$ coloring. Furthermore, when $\Delta=\sqrt{n}$, the results of \cite{CohenPW19,CohenW18} combine to show that no randomized online coloring algorithm can color using $\Delta+o(\sqrt{\Delta})$ colors. Thus, most of the focus has been designing algorithms for $\Delta=\omega(\log n)$ using $(1+o(1))\Delta$ colors.  There has been steady progress on this problem \cite{AggarwalMSZ03,BhattacharyaGW21,SaberiW21,KulkarniLSST22,BlikstadSVW23}, culminating in a recent work of \cite{BlikstadSVW24}, which presents a randomized algorithm that edge colors an online graph using $\Delta+o(\Delta)$ colors. 

In this paper, we investigate the \emph{randomized greedy algorithm}  which is a natural variation of the greedy algorithm. Given a set $\Gamma$ of colors, the randomized greedy algorithm $\mathcal{A}$ on input the online graph $G$ chooses the color of each arriving edge uniformly at random from the currently allowed colors for that edge, and leaves the edge uncolored if no colors are allowed. The algorithm is said to succeed if every edge is colored. The example in Figure~\ref{figure:tree} shows that the algorithm will fail if $|\Gamma|=\Delta+o(\sqrt{\Delta})$ (because the set of colors not
used by the first $\Delta-1$ edges is likely to be disjoint from the set
of colors not used by the second $\Delta-1$ edges.)
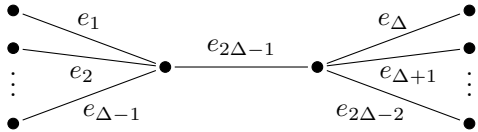
\begin{figure}[!ht]
\begin{center}\begin{tikzpicture}
    \filldraw (-2,.75) circle(2pt) node(w1)[]{};
    \filldraw (-2,.25) circle(2pt) node(w2)[]{};
    \filldraw (-2,-.75) circle(2pt) node(w3)[]{};
    \filldraw (0,0) circle(2pt) node(u)[]{};
    \filldraw (2,0) circle(2pt) node(v)[]{};
    \filldraw (4,.75) circle(2pt) node(x1)[]{};
    \filldraw (4,.25) circle(2pt) node(x2)[]{};
    \filldraw (4,-.75) circle(2pt) node(x3)[]{};
    \draw (-2,-.125) node()[]{$\vdots$};
    \draw (4,-.125) node()[]{$\vdots$};
    \draw (w1)--(u) node()[above, midway]{$e_1$};
    \draw (w2)--(u) node()[below, midway,xshift = -1mm]{$e_2$};
    \draw (w3)--(u) node()[below, midway, xshift = 3mm]{$e_{\Delta-1}$};
    \draw (u)--(v) node()[above, midway]{$e_{2\Delta-1}$};
    \draw (x1)--(v) node()[above, midway]{$e_{\Delta}$};
    \draw (x2)--(v) node()[below, midway,xshift = 2mm]{$e_{\Delta+1}$};
    \draw (x3)--(v) node()[below, midway,xshift = -3mm]{$e_{2\Delta-2}$};
\end{tikzpicture}\end{center}
\caption{A simple example to illustrate that if $\card{\Gamma}=\Delta+o(\sqrt{\Delta})$, then $\mathcal{A}$ likely fails.}
\label{figure:tree}
\end{figure}

In their paper, \cite{Bar-NoyMN92} conjectured that given any online graph $G$ with $\Delta=\omega(\log n)$ the randomized greedy algorithm $\mathcal{A}$ succeeds in $\Delta+o(\Delta)$ edge coloring $G$ with high probability. This algorithm is simple and natural, but
subsequent researchers (e.g. \cite{KulkarniLSST22,SaberiW21,amouzandeh2023sigact,BahmaniMM10})   have noted that it seems difficult to analyze.  Prior to the present paper, the  analysis was only done for the special case of trees (see \Cref{analysistrees} below). Our first result gives an analysis of this algorithm for the case of random-order arrivals. In this setting, an adversary can pick a worst-case graph but the edges of the graph arrive in a random order.  (Recall that a graph with no multiple edges
is said to be \emph{simple}.)

\begin{theorem2}[Informal version of \Cref{thm:randomordermain}]({\bf Random order case})
    Let $\epsilon>0$ be a constant. The algorithm $\mathcal{A}$, when given any simple graph $G$ of maximum degree $\Delta=\omega(\log n)$, whose edges are presented in a uniformly random order, edge colors $G$ with $(1+\epsilon)\Delta$ colors with high probability. 
\end{theorem2}

Our second result applies to the setting where the adversary is \emph{adaptive}. This means that the adversary does not have to decide the graph \emph{a priori}. They can instead decide the subsequent edges of graph based on the coloring of the prior edges. We call such a graph \emph{adaptively chosen}. 

\begin{theorem2}[Informal version of \Cref{thm:main}]\label{thm:denseinformal}({\bf Dense case})
Let $\epsilon>0$, $M \geq 1$ be any constants. Suppose $G$ is an \emph{adaptively chosen} simple graph with maximum degree $\Delta$ and $n \leq M\Delta$. Then, $\mathcal{A}$ succeeds in $(1+\epsilon)\Delta$-edge coloring $G$ with high probability. 
\end{theorem2}
To clarify the setting of \Cref{thm:denseinformal}, we note that  prior results (as far as we know) assumed an \emph{oblivious adversary} that chooses the graph $G$ and the edge-arrival order but must fix \emph{both} of these before their online algorithm receives any input. In contrast, an \emph{adaptive} adversary builds the input as $\mathcal{A}$ runs and may choose each edge depending on
the coloring of the graph so far.

For the case of online algorithms, \cite{Ben-DavidBKTW94} established a connection between randomized algorithms against an adaptive adversary and deterministic algorithms. Their result states that for a given online problem, if there is an $\alpha$-competitive randomized algorithm against an adaptive adversary, then there exists a $\alpha$-competitive deterministic algorithm. Exploiting this connection, we get the following corollary:

\begin{corollary2}
\label{cor:deterministic}
Let $\epsilon > 0, M  \geq 1$ be any constants.  For $\Delta$ sufficiently large and $n\leq M \Delta$ there is a deterministic online algorithm
that given 
any \emph{adaptively chosen} simple graph with $n$ vertices and maximum degree $\Delta$
will produce a valid coloring using $(1+\epsilon)\Delta$ colors.
\end{corollary2}

Since the above-mentioned previous algorithms for online edge-coloring were applicable against an \emph{oblivious adversary}, they did not have any implications to the deterministic setting. An independent and concurrent result, \cite{Blikstad24}, presents a deterministic online algorithm for bipartite edge-coloring under one-sided vertex arrivals that achieves competitive ratio $\frac{e}{e-1}+o(1)$ for all $\Delta = \omega(\log n)$. As far as we know, the result of \cite{Blikstad24} and \Cref{cor:deterministic} are the first results for a general class of graphs that give a deterministic online algorithm using fewer than  $2\Delta-1$ colors.

\paragraph{Related Work On Edge Coloring.}
The online edge coloring problem is closely related to the online matching problem. The relationship in one direction is quite clear: every color class in a proper edge coloring of a graph $G$ is a matching. In the other direction, \cite{SaberiW21} showed in order to solve the online edge coloring problem, it is sufficient to solve the online matching problem. More formally, their result states that if one can come up with an online matching algorithm that matches every edge with probability $1/\alpha\Delta$, then this can be turned into a $(\alpha+o(1))$-competitive online edge coloring algorithm. Several recent papers \cite{CohenPW19,SaberiW21,BlikstadSVW23,BlikstadSVW24,KulkarniLSST22}, with the exception of \cite{BhattacharyaGW21} have exploited this connection, and designed online matching algorithms which match each edge with some suitable probability. In contrast, in this paper, our focus is on the randomized greedy algorithm and we make some progress towards analysing this natural random process. Since the result of \cite{BhattacharyaGW21} is also in the random order setting but analyses a different algorithm, we discuss that further in \Cref{sec:ourapproach}.

\paragraph{Edge Coloring in Other Computational Models.}
The edge coloring problem has been considered in numerous settings. In the standard (offline) computational model, it is NP-hard to distinguish whether a graph is $\Delta$ or $\Delta+1$ colorable \cite{Holyer81a}, so attention has focused on algorithms that color in at most $\Delta+1$ colors.
As mentioned before, Vizing's proof 
gives such an  algorithm that runs in $O(mn)$ time. This was  improved to $O(m\cdot \min\set{\Delta\cdot \log n, \sqrt{n\cdot \log n})}$  by \cite{Arj82,GabowNKLT} and then to $O(m\cdot \min\set{\Delta\cdot \log n, \sqrt{n})}$ by \cite{Sinnamon19}. These bounds were improved in a series of works \cite{BhattacharyaCCSZ24,Assadi24,BhattacharyaCSZ24}, culminating in \cite{AssadiBBCSZ24} which gives an $O(m\log \Delta)$ time algorithm for $\Delta+1$ edge coloring a graph. There has also been a considerable effort to study edge coloring in other computational models: such as the dynamic model \cite{BarenboimM17,BhattacharyaCHN18,DuanHZ19,Christiansen23,BhattacharyaCPS24,Christiansen24}, the distributed model \cite{PanconesiR01,ElkinK24,FischerGK17,GhaffariMU18,BalliuBKO22,ChangHLPU20,Bernshteyn22,Christiansen23,Davies23}, and the streaming model \cite{BehnezhadDHKS19,BehnezhadS23,ChechikDZ23,GhoshS23}.

\subsection{Our Approach}
\label{sec:ourapproach}
We start with describing $\mathcal{A}$ in more detail. Fix the color set $\Gamma$.  
The edges of $G=(V,E)$ arrive in some order $e_1,e_2,...$ and so on. 
At all points during the algorithm,
for each vertex $v$, the \emph{free set at $v$} is the set of colors
not yet used to color an edge incident on $v$.  The free set of $v$ when
edge $e_i$ arrives is denoted $F_{i-1}(v)$. When edge $e_i = (u,v)$ arrives, algorithm $\mathcal{A}$ uniformly chooses a color $c$ from $F_{i-1}(u) \cap F_{i-1}(v)$ and colors $e_i$ with that color. If $F_{i-1}(u) \cap F_{i-1}(v) = \emptyset$ when $e_i$ arrives then $e_i$ is left uncolored and the algorithm continues.  In this case the algorithm is said to \emph{fail on $e_i$}.

As mentioned in the introduction, this algorithm was previously analyzed
for trees. Since the tree case gives intuition for our proof strategy for general graphs, we give a proof of this theorem for completeness.

 \begin{theorem2}\label{analysistrees}\cite{FederMP}
 Suppose the online graph $G$ is a tree and $\Delta=\omega(\log n)$. If $\card{\Gamma}=\Delta+2\sqrt{\Delta\cdot \log n}$, then $\mathcal{A}$ succeeds with probability at least $1-1/n$.
\end{theorem2}

\begin{proof}

We prove an upper bound on  the probability
that the algorithm fails to color the $i$th edge $e_i=(u,v)$ conditioned on not having
failed previously.   Let $d_{i-1}(u)$, respectively $d_{i-1}(v)$, be the degree of $u$ and $v$ prior
to the arrival of $e$ and let $k(u)=(1+\epsilon)\Delta-d_{i-1}(u)$ and $k(v)=(1+\epsilon)\Delta-d_{i-1}(v)$. Conditioned on the algorithm having succeeded thus far, $\card{F_{i-1}(u)}=k(u)$ and $\card{F_{i-1}(v)}=k(v)$.  By symmetry $F_{i-1}(u)$ and $F_{i-1}(v)$ are uniformly random from $\binom{\Gamma}{k(u)}$, $\binom{\Gamma}{k(v)}$, respectively.  Since $u$ and $v$ are in separate components of the tree prior to the arrival of $e_i$, $F_{i-1}(u)$ and $F_{i-1}(v)$ are independently sampled from $\Gamma$, even if we condition on successful coloring prior to $e_i$. Therefore, conditioned on  successful coloring  prior to $e_i$, 
\begin{align*}
\prob{\text{$e_i$ not colored}}&\le \prob{F_{i-1}(u)\cap F_{i-1}(v)=\emptyset}\\
& =\sum_{S \in \binom{\Gamma}{k(u)}}\prob{S\cap F_{i-1}(v)=\emptyset\mid F_{i-1}(u)=S}\cdot \prob{F_{i-1}(u)= S}&\\
& \overset{(\ref{independence})}{=} \sum_{S \in \binom{\Gamma}{k(u)}}\prob{S \cap F_{i-1}(v) = \emptyset}\cdot \prob{F_{i-1}(u)= S}\\
&\leq \sum_{S \in \binom{\Gamma}{k(u)}}\paren{1-\frac{k(u)}{\card{\Gamma}}}^{k(v)}\cdot \prob{F_{i-1}(u)= S}\\
&\leq\exp\paren{-\frac{k(u)k(v)}{|\Gamma|}}\\
&\overset{(\ref{kukv})}{\leq} \exp\paren{-\frac{\epsilon^2\cdot \Delta}{(1+\epsilon)}}\\
&\overset{(\ref{epsvalue})}{=} O(n^{-3})
\end{align*}
where
  \begin{inparaenum}[$(i)$]
    \item\label{independence}follows from independence of $F_{i-1}(u)$ and $F_{i-1}(v)$,
    \item\label{kukv}follows from $k(u)\geq \epsilon\Delta$ and $k(v)\geq \epsilon\Delta$, and 
    \item\label{epsvalue}is implied by the fact that $\epsilon\geq 2\sqrt{\frac{\log n}{\Delta}}$.
  \end{inparaenum}
Taking a union bound over all edges $e_i$, $\mathcal{A}$ 
fails with probability $O(n^{-1})$.
\end{proof}

Taking inspiration from the proof for trees, we aim to show that for any edge $(u,v)$, the sets $F_{i-1}(u)$ and $F_{i-1}(v)$, though not independent, are close to independent.  It is difficult to estimate the dependencies of the sets $F_{i-1}(u) \cap F_{i-1}(v)$ directly. Instead we 
aim to show that for {\em any} fixed  subset $S$ of colors, and for any time step $i$, 
the set $F_{i-1}(v)$ ``looks  random'' with respect to $S$: 
\begin{equation}\label{eqn:independenceS}
|S \cap F_{i-1}(v)| \approx \frac{|S||F_{i-1}(v)|}{(1+\epsilon)\Delta}.
\end{equation}
The right hand side is the expected size of $|S \cap F_{i-1}(v)|$
if $F_{i-1}(v)$ was a uniformly random subset of $\Gamma$ of size $|F_{i-1}(v)|$.
It is fairly obvious that \Cref{eqn:independenceS} cannot hold for every set $S$,  vertex $v$, and
time step $i$ (for example, $S=\Gamma-F_{i-1}(v)$).
We will show that for each $S$, and for all but constantly many vertices $v$, \Cref{eqn:independenceS} holds for all time steps $i$ (up to some small error).

The intuition for our proof is as follows.  The set  $F_i(v)$ evolves over time as colors are removed one by one. If, at each step $j$, the color removed was chosen uniformly from $F_{j-1}(v)$, then at any time $i$, the set $F_i(v)$ would be a uniform set from $\Gamma$ and would likely satisfy \Cref{eqn:independenceS}. In fact, since we are only interested in the {\em size} of $F_i(v) \cap S$ appearing to be random, and not the set itself, in order to satisfy \Cref{eqn:independenceS} with high probability, it is enough that the probability the color chosen at step $j$ is in $S$ is the same as it would be if the color was chosen uniformly from $F_{j-1}(v)$. More explicitly, we would like to have
\begin{align}\label{eqn:intersectionchange}
\frac{|S \cap F_{j-1}(v) \cap F_{j-1}(w)|}{|F_{j-1}(v) \cap F_{j-1}(w)|} \approx \frac{|S \cap F_{j-1}(v)|}{|F_{j-1}(v)|},
\end{align}
where the left hand side is the probability that the color selected for $e_j=(v,w_j)$ belongs to $S$. 
Note that if we set $S_j'= S \cap F_{j-1}(v)$ and $S_j'' = F_{j-1}(v)$, and in addition we knew that \Cref{eqn:independenceS} was satisfied for vertex $w_j$ and each of the color sets $S'$ and $S''$,  then \Cref{eqn:intersectionchange} would hold (with a marginal increase in error.)  This idea forms the basis of an inductive proof of our main lemma, \Cref{lem:mainlemma}.

As mentioned,  we will show that with high probability, for each set $S$,  the number of vertices $v$ for which
\Cref{eqn:independenceS} fails is bounded by a constant. Ideally, we would like to use this to show that for all but constantly many $e_j = (v,w_j)$ adjacent to $v$, \Cref{eqn:intersectionchange} holds. However, note that the set whose intersection we are trying to control in \Cref{eqn:intersectionchange} is a moving target: for each neighbor $w_j$ of $v$, we would like its intersections with the free set $S_j^{''} = F_{j-1}(v)$ and $S_j^{'}= S \cap F_{j-1}(v)$  to be what we expect. Unfortunately, this does not exclude the possibility that for every $j$, $w_j$ is one of the exceptional vertices for which \Cref{eqn:independenceS} fails for $S_j^{'}$ or $S_j^{''}$. This would prevent us from reasoning that exceptions to \Cref{eqn:intersectionchange} are only constantly many neighbors of $v$. 

To get around this we note that for fixed $v$, as $j$ varies the sets $F_j(v)$ are closely related - $F_j(v)$ will differ from $F_{j+t}(v)$ by at most $t$ colors. So for each vertex $v$, 
we will partition the time steps into a (large) constant number of \emph{$v$-phases} so that
each $v$-phase contains only a small fraction of the edges incident on $v$. Then for each $v$-phase $r$
we will approximate $F_{j-1}(v)$ for all time steps $j$ of that $v$-phase by the set $A_{j-1}(v)$
which is defined to be the free set of $v$ at the beginning of the $v$-phase that contains $j$. Since $F_{j-1}(v)$ will not differ too much from $A_{j-1}(v)$, we will argue that if we replace $F_{j-1}(v)$ with $A_{j-1}(v)$ in \Cref{eqn:intersectionchange} and the new equation holds for all but constantly many $e_j = (v,w_j)$, then the original equation must also hold for all but constantly many neighbors $w_j$. This allows us to carry out the reasoning of the previous paragraph. That is, now we have constantly many possibilities for the $S'_j=A_{j-1}(v)\cap S$ and $S''_j=A_{j-1}(v)$, and for each of these possibilities, constantly many vertices do not satisfy \Cref{eqn:independenceS}. Thus, for all except constantly many neighbours $w_j$ of $v$, we can hope for \Cref{eqn:intersectionchange} to hold.

To this end, our analysis will consider a modified version of the algorithm which we denote by $\mathcal{A}'$.
This modified version produces the same distribution over colorings as $\mathcal{A}$ but
it has the advantage that it explicitly reflects the partition of time steps into $v$-phases for each $v$,
and the approximation of $F_{j-1}(v)$ by $A_{j-1}(v)$
described above.
The modified algorithm works as follows.  When edge $e_j=(u,v)$ arrives
we first sample a color uniformly from $A_{j-1}(u) \cap A_{j-1}(v)$ (rather than
from $F_{j-1}(u) \cap F_{j-1}(v)$).  If the selected color is valid (which means
that it belongs to $F_{j-1}(u) \cap F_{j-1}(v)$) then we use it, otherwise we discard it
and then sample uniformly from $F_{j-1}(u) \cap F_{j-1}(v)$. We will have the following chain of comparisons: first, that the number of times $\mathcal{A}'$ initially selects an invalid color and resamples is small (a constant fraction of times), and thus the colors chosen for $\mathcal{A}$ are close to the colors initially chosen by $\mathcal{A}'$. Second, the inductive assumption that \Cref{eqn:intersectionchange} holds for all but constantly many edges $e_j = (v,w_j)$ adjacent to $v$, and thus the colors initially chosen by $\mathcal{A}'$ hit any given set $S$ about as often as a uniform choice from $F_{j-1}(v)$ would. Finally, we show that the uniform choice from $F_{j-1}(v)$ would likely satisfy \Cref{eqn:independenceS}. \Cref{prop:deltatomartingales} quantifies the error added in each step of this chain in terms of martingale sums defined in \Cref{section:differencesequences}.

\paragraph{Comparison to Prior Work.}
Despite the random greedy algorithm being the most natural candidate to consider for this problem, we are not aware of any prior analysis its behavior. Here we discuss the unique challenges of analyzing this algorithm in comparison to prior work in this area, and the ideas we use to address these challenges.

Prior work has focused on designing and analyzing other algorithms.
The fundamental difference between the random greedy algorithm and previously studied algorithms for online edge coloring is the total dependence of each step on all previous steps of the algorithm. Although at first glance, our algorithm $\mathcal{A}'$ may appear similar to previous algorithms such as that of \cite{BhattacharyaGW21} or \cite{KulkarniLSST22} in its use of phases, the phases in these works play a fundamentally different role in their analyses than they do here. 

One key difference is that these previous algorithms initially split the color set into palettes, and the algorithm specifies which palette to select the color from for each edge.
In \cite{BhattacharyaGW21}, each phase represents a window during which the algorithm enforces that edges make independent choices about their preliminary colors. Each edge participating in a given phase picks a preliminary color from the palette. These choices are made independently. If the color chosen by an edge is unique in its neighborhood, then this is made permanent. Otherwise, an edge is considered to have failed. All failed edges are colored using a separate palette. On the other hand, in \cite{KulkarniLSST22}, phases are used to reduce an instance of online edge coloring to independent instances of online edge coloring on trees (which is considerably simpler). Different palettes are used for different phases.

In contrast, under the simple greedy strategy of algorithm $\mathcal{A}$, dependencies between different choices of edges are inherent and pervasive.  We introduce algorithm $\mathcal{A}'$  (which produces the exact
same distribution on colorings as $\mathcal{A})$ as an analytical tool for
analyzing $\mathcal{A}$.  The phase structure of $\mathcal{A}'$ does not eliminate dependencies within or between phases but does allow us to bound their impact.

Algorithm $\mathcal{A}'$ is restricted to using the same palette throughout the course of the algorithm; 
the role of a phase is purely analytical and is used to allow us to carry out the inductive argument described before.
Although the initial color for every edge in the phase is chosen independently from the palette, once a collision occurs, the final colors used are dependent on the previous choices made during the phase. It is not too hard to see that collisions end up occurring a constant fraction of the time. Therefore, even within a phase a large fraction of our choices are dependent. Thus, it is no longer clear that the colors removed from a palette in a phase appear consistent with an independent coloring during that phase. These dependencies make it difficult to analyze the algorithm using standard tools. 

In order to address this, we devise a metric, $\delta$, to measure how ``independent" color sets appear with respect to each other. We then decompose this metric into two main terms: one which reflects the natural variance in independent choices of color sets, and one which accounts for the build up of error due to the restrictions of the algorithm. We use martingale tail bounds to control the behavior of the first term,  These bounds apply in the worst case setting provided that $\Delta=\omega(\log n)$,  despite the
dependencies between steps. To control the second term we use an inductive argument, and it is here that
we require some additional assumptions, either random order (Proposition \ref{prop:boundonepsilonrandom}) or
adverserial order but sufficiently dense (Proposition \ref{prop:boundonepsilondense}).
These ideas provide a framework for understanding the behavior of a deeply dependent process and comprise the main technical contribution of this paper.

\paragraph{Paper Outline.} The paper is be outlined as follows. In \Cref{sec:prelims}, we state give notation and state our results formally. Additionally, we give more intuition about our proof in \Cref{sec:proofidea}. In \Cref{section:errorbounds}, we give crucial properties of the phases we described in \Cref{sec:ourapproach}. Our main tool in making the argument in \Cref{sec:ourapproach} will be martingales, and in \Cref{sec:martingales} we will prove a novel martingale concentration lemma (\Cref{lem:generalconcentration}) that will be crucial to our analysis. In \Cref{section:wellbehavedcolorings} we describe some important properties of colorings induced by our algorithm $\mathcal{A}$ and in \Cref{section:mainlemma} we prove our main lemma. In \Cref{sec:FurtherWork}, we give some extensions: \Cref{subsection:improvement}, we show strengthened versions of our main theorems (\Cref{thm:randomordermain} and \Cref{thm:main}). \Cref{subsection:densecase} discusses the degree assumption in \Cref{thm:main}. It also presents some starting points for improving the degree assumption. Finally, in \Cref{subsection:StaticAlgo} we also give an implementation of the randomized greedy algorithm in the static setting that has a runtime of $m\cdot 2^{O(1/\epsilon^2)}$ with high probability.

\section{Preliminaries}
\label{sec:prelims}

\subsection{Online Coloring}
Online coloring can be described as a two player game between \emph{Builder} and \emph{Colorer}.  The game $\game(n,\Delta,\epsilon)$ is parameterized by the number of vertices $n$, degree
bound $\Delta$ and $\epsilon>0$.   The game starts with the empty graph on $n$ vertices and a color
set $\Gamma$ of size
$\lceil(1+\epsilon)\Delta\rceil$. The game lasts for $m=\lfloor\Delta n/2\rfloor$ steps. In each time step, Builder 
selects an edge to add to the graph, subject to the restriction that
all vertex degrees remain below $\Delta$. Colorer then assigns a color to the edge
from $\Gamma$ so that the overall coloring remains valid.  If this is impossible (i.e.
any color choice will invalidate the coloring) the edge is left uncolored. 
Colorer wins
if every edge  is successfully colored. For convenience, we also allow Builder to add {\em null edges}; such edges are not adjacent to any vertex in the graph and may be assigned any color without affecting whether or not the coloring is valid. This has the advantage that we can 
fix the  number of steps to $m$, and if Builder gets stuck (is unable to add
an edge without violating the degree bound) then he can add null edges for the remaining steps.

The \emph{state} of the game after $i$ steps, denoted $\mathcal{S}_i$, consists of the list of the first
$i$ edges chosen by Builder and the coloring chosen by Colorer.  A strategy for Builder is a function which given any game state
$\gamestate_i$ determines the next edge to be added or terminates the game.  A strategy for Colorer
is a function which given the game state $\gamestate_i$ and an additional edge $e$ assigns a color to $e$ (or leaves $e$ uncolored if no color can be assigned.)

We say that a color $c$ is \emph{free for $v$ at step $i$} if among the first $i-1$
edges, no edge that touches $v$ is colored by $c$.
If the $i$th edge chosen by Builder is $(u,v)$ then the set of colors that can be used
to color $(u,v)$ 
is the intersection of the free set for $v$ at step $i$ and the free set
for $u$ at step $i$. If that set is empty then the edge is necessarily left uncolored. 

We are interested in analyzing the behavior of the  
\emph{randomized greedy strategy} for Colorer, denoted $\mathcal{A}$: for each new edge $e_i=(u,v)$,
if the intersection of the free set for $u$ and the free set for $v$ is nonempty then
choose the color for $e_i$ uniformly at random from that set.

During the game Builder produces a graph $G$ together with an ordering of its edges
which we view as a one-to-one function $\sigma:E(G) \longrightarrow [m]$. Elements in $[m]$ to which no edge is in $E(G)$ is mapped are interpreted as \emph{null edges}.
We refer to $(G,\sigma)$ as an edge-ordered graph.  In general the edge
ordered graph produced by Builder may depend on the coloring of edges chosen
by Colorer.
A strategy of  Builder is \emph{oblivious} if the choice of edge to be added at each step
depends only on the current edge set  and not on the coloring. An oblivous strategy is
fully described by the pair $(G,\sigma)$ where the edge selected at step $i$ is $\sigma^{-1}_i$ (and
is a null edge if that is undefined).
We denote this strategy by $\obl(G,\sigma)$. We can now give a more precise formulation of our first main result.

\begin{restatable}{theorem2}{mainthmrandom}\label{thm:randomordermain}({\bf random order case})
For any constant $\epsilon\in (0,1)$ there are constants $N=N(\epsilon)$ (sufficiently large), $\gamma_1=\gamma_1(\epsilon)$ and $\gamma_2=\gamma_2(\epsilon)$ (sufficiently small) such that the following holds.  Suppose that $n$ is sufficiently large, $\Delta > N\log(n)$ and consider the edge coloring game 
$\game(n,\Delta,\epsilon)$. For any $G$ on $n$ vertices with maximum degree $\Delta$,
for all but at most a $2^{-\gamma_{1} \Delta}$ fraction of mappings $\sigma$ of $E(G)$ to $[m]$, $\mathcal{A}$  will defeat the oblivious strategy $\obl(G,\sigma)$ (i.e, produce an edge coloring for $G$)  with probability at least $1-2^{-\gamma_{2}\Delta}$.
\end{restatable}

Our second main result applies to arbitrary adaptive strategies of Builder. (Here we use
the word \emph{adaptive} to emphasize that Builder's choices may depend on the past coloring.)

\begin{restatable}{theorem2}{mainthmadversarial}\label{thm:main}({\bf dense case}) For any constant 
$\epsilon\in (0,1)$ and constant $M>1$, there is a constant $\gamma=\gamma(\epsilon,M)$ so that
the following holds. For sufficiently large $n$, and for $\Delta>n/M$, for any
(possibly adaptive) strategy
for the online coloring game $\game(n,\Delta,\epsilon)$, $\mathcal{A}$ wins (produces an edge
coloring of the resulting graph)
with probability $1-2^{-\gamma\Delta}$.
\end{restatable}

In \Cref{sec:FurtherWork} we show that \Cref{thm:main} and \Cref{thm:randomordermain} can be slightly improved. The formal statements of improved versions are stated in \Cref{lem:newdensecase} and \Cref{lem:randomordernew}, respectively.

\subsection{The Algorithm \(\mathcal{A}'\)}

In order to prove the above theorems we analyze a modified strategy $\mathcal{A}'$ that against
any given Builder strategy produces exactly the same distribution over colorings as $\mathcal{A}$, but will be
easier to analyze.  For each vertex $v$, the modified algorithm will keep track of a partition
of the time steps into at most $b=b(\epsilon)$ contiguous \emph{$v$-phases}. The value of $b(\epsilon)$ is specified in \Cref{section:badevents}.

For each $v$, the $v$-phases  are numbered from 1 to $b$. 
The partition into $v$-phases is represented by a \emph{phase-partition counter} $\{\phi_i(v):i \in \{1,\ldots,t\}\}$ where $\phi_i(v)$ is the number of the $v$-phase that contains time step $i$. 
For each $i>1$ we have either $\phi_i(v)=\phi_{i-1}(v)$ (if $i-1$ and $i$ are in the same phase) or $\phi_i(v)=1+\phi_{i-1}(v)$ (if $i$ starts a new $v$-phase.)
This function is determined online, so that $\phi_i(v)$ is determined after
step $i-1$ of the game.

The phase-partition counters that we use are defined formally in \Cref{def:phasepartitionfunc}.
In the dense case (\Cref{thm:main}) for each vertex $v$, the $v$-phases  are determined
by the number of edges incident to $v$ that have arrived. For $r\geq 2$, the $r^{th}$ $v$-phase starts
with the time step where the number of edges incident
on $v$ first exceeds $(r-1)\Delta/b$. (Thus the number of edges incident on
$v$ in each $v$-phase is within 1 of $\Delta/b$.)  The phase-partition counters in this case
are denoted by $\phidense$.

In the random case (\Cref{thm:randomordermain}), the $v$-phase partition is the
same for every vertex. For $r\geq 2$, the $r^{th}$ $v$-phase starts
with the time step where the total number of edges arrived first exceeds $(r-1)m/b$. The phase-partition counters
in this case are denoted $\phirandom$.

\paragraph{Algorithm Description}In $\mathcal{A}'$, Colorer maintains for each vertex $v$, a color set $A_i(v)$ that approximates
$F_i(v)$ but remains constant during each $v$-phase. We call $A_i(v)$ the {\em palette} of $v$ at the end of time step $i$. For vertex $v$ and time $i$, we define $A_i(v)=\Gamma$ if $i+1$ belongs to the first $v$-phase and otherwise:
\begin{eqnarray*}
A_{i}(v)
& = & F_{i'}(v) \text{ where $i'$ is the final step of the $v$-phase prior to the $v$-phase containing step $i+1$.}\\
&= &\text{the set of available (free) colors $v$ at the completion 
of $v$-phase $\phi_{i+1}(v)-1$.}
\end{eqnarray*}

Note that $A_i(v) \supset F_i(v)$ for all $i$, and we think of $A_i(v)$ as an (over-)approximation
to $F_i(v)$.

\begin{definition2}[Algorithm $\mathcal{A}'$]\label{def:modifiedalg}
    Start with $A_0(v) = F_{0}(v)=\Gamma$ for all $v$. When edge $e_i = (u,v)$ arrives:
    \begin{enumerate}
        \item Choose $c$ uniformly at random from $A_{i-1}(u) \cap A_{i-1}(v)$. This is the \emph{preliminary color for $e_i$}. If $A_{i-1}(u) \cap A_{i-1}(v) = \emptyset$,  $e_i$ is left uncolored.
        \item Next, choose the final color for $e_i$:
        \begin{enumerate}[label={\roman*)}]    
            \item If $c \in F_{i-1}(u) \cap F_{i-1}(v)$, color edge $e_i$ with $c$.
            \item Otherwise, $c \notin F_{i-1}(u) \cap F_{i-1}(v)$. (We refer to this as a \emph{collision} at $e$.) In this case, choose $c'$ uniformly from $F_{i-1}(u) \cap F_{i-1}(v)$ for edge $e_i$. 
            If $F_{i-1}(u) \cap F_{i-1}(v) = \emptyset$, $e_i$ is left uncolored.  \label{alg:resamplestep}
    
        \end{enumerate}
        \item For all vertices $w$, if $\phi_i(w) < \phi_{i+1}(w)$ ($i$ completes the current $w$-phase) then $A_{i}(w)$ is set to $F_i(w)$, otherwise $A_{i}(w)=A_{i-1}(w)$.
    \end{enumerate} 
\end{definition2}

It is obvious from the algorithm that the final color selected for $e$ is uniformly random
over  $F_{i-1}(u) \cap F_{i-1}(v)$, so the distribution over colorings produced by $\mathcal{A}'$
is the same as $\mathcal{A}$. To prove  Theorems~\ref{thm:randomordermain} and~\ref{thm:main} it suffices to prove the corresponding statements with $\mathcal{A}$ replaced by $\mathcal{A}'$ and this is what we'll do. 

\subsection{The probability space}

Fix any (possibly adaptive) strategy for Builder and consider the coloring game
against $\mathcal{A}'$.  Recall that the state of the game $\gamestate_i$ after $i$ steps
consists of the sequence $e_1,\ldots,e_i$ of edges and the coloring of those edges.
For the purpose of analysis we augment $\gamestate_i$ so that
it includes the preliminary color that $\mathcal{A}'$ chooses for each edge. The random choices of $\mathcal{A}'$ determine a discrete
probablity space over the space of final states of the game;
we view this space as a stochastic process that evolves over time.

\paragraph{Notation.}

For a step $1 \leq i \leq m$ we define the following random variables depending only on $\mathcal{\gamestate}_i$, the state after the $i^{th}$ edge is processed: 
\begin{eqnarray*}
e_{i+1} &=& \text{the $(i+1)^{th}$ edge to arrive.}
\end{eqnarray*}
For vertex $v$ and time step $i$,  we define the following random variables that are determined by  $\mathcal{\gamestate}_i$: 
\begin{eqnarray*}
F_i(v) &=& \text{the set of free colors of $v$ at the end of step $i$.}\\
\phi_{i+1}(v) &=& \text{the $v$-phase for step $i+1$.}\\
A_i(v) &=& \text{the palette of free colors of $v$ at the end of step $i$.}
\end{eqnarray*}

We will track the above variables as the algorithm progresses, and they will form the basis for the martingale difference sequences we analyze in \Cref{section:wellbehavedcolorings}. Each of these variables is subscripted by a step $i$, at which time they are fixed. This is what allows us to apply our concentration lemmas and make probabilistic claims about such variables.

In contrast, the variables defined below will be superscripted by a phase number $r$, if at all, and are defined only after the last edge in phase $r$ of $v$ arrives and is colored (or fails to be colored.) However, since the choice of edges is not fixed ahead of time, we do not know a priori the step $i$ at which they will be defined. Thus, we must be careful when dealing with such quantities. In the rest of the paper, outside of \Cref{section:wellbehavedcolorings}, we will assume the algorithm has concluded, and make deterministic statements about the quantities defined below, conditioned on the results from \Cref{section:wellbehavedcolorings}.

For a vertex $v$ the following set depends on the final state $\mathcal{\gamestate}_{\final}$:
\begin{eqnarray*}
\timesteps(v)
&=& \text{ the set of arrival times of edges adjacent to $v$.}
\end{eqnarray*}
Also for vertex $v$ and $v$-phase $r$, we define the following sets depending on $\mathcal{\gamestate}_{\final}$:
\begin{eqnarray*}
A^r(v) &=& \text{the set of free colors of $v$ at the end of the $r^{th}$ $v$-phase.}\\
{U}^r(v) &=& A^{r-1}(v) \setminus A^{r}(v)\\
&=& \text{ the set of colors used to color edges incident on $v$ during the $r^{th}$ $v$-phase.}\\
\timesteps^r(v) &=& \{i \in \timesteps(v): \phi_i(v) = r\}\\
& = & \text{ the set of arrival times of edges incident on $v$ during the $r^{th}$ $v$-phase.}\\
\timesteps^{\leq r}(v) &=& \cup_{r' \leq r} \timesteps^{r'}(v).\\
\last{r}{v} &=& \max\set{i: i\in \timesteps^{r}(v)}\\
&=& \text{ the time at which the last edge of $v$-phase $r$ that is adjacent to $v$ arrives}.
\end{eqnarray*}
We remark that by this notation, for each time $i$, $A_{i-1}(v) = A^{\phi_i(v)-1}(v)$. 

Finally, we define the following quantity, called the \emph{error of vertex $v$ with respect to $S$ after phase $r$}:
\begin{align}\label{def:deltavS}
\delta^r(v,S) :=  \frac{|A^r(v) \cap S|}{|A^r(v)|} - \frac{|S|}{(1+\epsilon)\Delta}.
\end{align}

In contrast, the variables defined below will be superscripted by a phase number $r$, if at all, and are defined only after the last edge in phase $r$ of $v$ arrives and is colored (or fails to be colored). Thus, the times $\last{r}{v}$ are important to note down. More formally, we define the following partial order on vertex-phase pairs $(v,r) \in V \times [b]$.

\begin{definition2}[Vertex-Phase and Edge Arrival Ordering, $\prec$]\label{def:partialorder}
Given pairs $(v,r), (u,s) \in V(G) \times [b]$, we say that $(v,r) \prec (u,s)$ if $\last{r}{v} < \last{s}{u}$.
 \end{definition2}

We make one final note about the variables defined in this section. Above we have defined two families of variables: those indexed by times steps $i$ that are defined at a known time step $i$, and those that are not defined until the algorithms concludes. Note that although the sets in the second family depend on $\mathcal{\gamestate}_{\final}$, the indicator variable for a particular index $i$ belonging to one of these sets only depends on $\mathcal{\gamestate}_{i-1}$, since the identity of $e_i$ and $\phi_i(v)$ are determined at that point. For instance, $T^r(v)$ depends on $\mathcal{\gamestate}_{\final}$, but the indicator $\mathds{1}_{\{i \in T^r(v)\}}$ only depends on $\mathcal{\gamestate}_{i-1}$. Thus, indicator variables of this form will also belong to the first family of variables and will be used in \Cref{section:wellbehavedcolorings}.

\subsection{Proof Idea}\label{sec:proofidea}
As mentioned above, for each set $S \subseteq \Gamma$ and vertex $v \in V$, we would like to track $|F_{i}(v) \cap S|$ as $\mathcal{A}'$ progresses. We will do this indirectly by tracking
$\frac{|A^r(v) \cap S|}{|A^r(v)|}$ with the goal of bounding  $\delta^r(v,S)$. Observe that $\delta^0(v,S)=0$, so our goal will be to bound how much $\delta^r(v,S)$ increases in each phase. Note that if ${U}^r(v)$ was chosen uniformly from $A^{r-1}(v)$, we would expect that
\begin{align}\label{eqn:usedcolors}
    |{U}^r(v) \cap S| \approx |{U}^r(v)|\cdot \frac{|A^{r-1}(v) \cap S|}{|A^{r-1}(v)|},
\end{align}
which would imply that the error with respect to any set $S$ would not increase too much when the phase of $v$ changes:
\begin{align}
    \delta^{r-1}(v,S) - \delta^{r}(v,S) & =  \frac{|A^{r-1}(v) \cap S|}{|A^{r-1}(v)|}- \frac{|A^{r}(v) \cap S|}{|A^{r}(v)|}\\
    &=  \frac{|A^{r-1}(v) \cap S|}{|A^{r-1}(v)|}- \frac{|A^{r-1}(v) \cap S|-|{U}^{r}(v) \cap S|}{|A^{r}(v)|} \nonumber \\ 
    & = \frac{1}{|A^{r}(v)|}\left(|{U}^{r}(v) \cap S| +|A^{r-1}(v) \cap S|\left(\frac{|A^{r}(v)|}{|A^{r-1}(v)|}-1\right)\right)  \nonumber \\
    & = \frac{1}{|A^{r}(v)|}\left(|{U}^{r}(v) \cap S| + |A^{r-1}(v) \cap S|\cdot \frac{|A^{r}(v)|-|A^{r-1}(v)|}{|A^{r-1}(v)|}\right)  \nonumber \\
    & = \frac{1}{|A^{r}(v)|}\left(|{U}^{r}(v) \cap S| - |A^{r-1}(v) \cap S|\cdot \frac{|{U}^{r}(v)|}{|A^{r-1}(v)|}\right) \label{eqn:usedcolorssize}.
\end{align}

Our goal is to show that \Cref{eqn:usedcolors} is never too far from the truth, and therefore $\delta^r(v,S)$ does not grow too large in any given phase. As we previously noted, this cannot hold for every set $S$ and every vertex $v$, but we can show that there exists a constant $C$ depending on $\epsilon$ such that with high probability, for each set $S$, for all but at most $C$ vertices $v$ and all phases $r$, we have
\begin{align}\label{eqn:FSetIntersecIndependent}
    \left|\delta^r(v,S) \right|\leq \frac{\epsilon^3\Delta}{10|A^{r}(v)|}.
\end{align} 
\Cref{prop:deltatomartingales} allows us to bound the amount $\delta^r(v,S)$ grows during a phase of $v$ in terms of three main sources of error: the error from the collisions at each phase, the error inherent to a locally independent algorithm, and the error in the palettes of the neighbors of $v$ for that phase. We will formalize these errors in the subsequent paragraphs. 

To that end, for all $e_i = (u,v)$, we define the indicator variables $X_{i}(S)$ to track whether the preliminary color chosen for $e_i$ from $A_{i-1}(u) \cap A_{i-1}(v)$ hits $S$, $Y_{i}(S)$ to track whether the final color chosen for $e_i$ from $A_{i-1}(u) \cap A_{i-1}(v)$ hits $S$, and the collision indicator variables $Z_i$ to track whether the preliminary color chosen for $e_i$ needs to be resampled. Furthermore, we let \begin{align}\label{eqn:collisionmartingale}
p_{i}(S) := \frac{|A_{i-1}(u) \cap A_{i-1}(v) \cap S|}{|A_{i-1}(u) \cap A_{i-1}(v)|}
\end{align}
be the probability that $X_{i}(S) = 1$, conditioned on the partial coloring of edges before $e_i$ arrives, and let 
\begin{align}
    D_i(S) = X_i(S)-p_i(S).
\end{align}
Note that,
\begin{align}\label{eqn:relXZ}
    \left||U^r(v) \cap S|-\sum\limits_{j \in \timesteps^r(v)} X_{j}(S)\right| = \left|\sum\limits_{j \in \timesteps^r(v)} Y_{j}(S) - X_{j}(S)\right| \leq \sum\limits_{j \in \timesteps^r(v)} Z_{j}
\end{align}
since the colors used in the preliminary coloring and the final coloring differ only in the edges which experience collisions. This produces the first source of error. Similarly,
\begin{align}\label{eqn:relXp}
    \left|\sum\limits_{j \in \timesteps^r(v)} X_{j}(S)-\sum\limits_{j \in \timesteps^r(v)} p_{j}(S)\right| = \left|\sum\limits_{j \in \timesteps^r(v)} D_{j}(S)\right|
\end{align}
models the inherent error of the local algorithm on phase $r$ of $v$. Therefore, if we can show that for most $i \in \timesteps^r(v)$,
\begin{align}\label{eqn:approximatep}
    p_{i}(S) \approx \frac{|A^{r-1}(v) \cap S|}{|A^{r-1}(v)|}
\end{align}
and bound the quantities on the right hand sides of \Cref{eqn:relXZ} and \Cref{eqn:relXp},
we will have
\begin{align}\label{eqn:approximateu}
     |U^{r+1}(v) \cap S| \approx |U^{r+1}(v)|\cdot \frac{|A^{r}(v) \cap S|}{|A^{r}(v)|},
\end{align}
as desired. Here, \Cref{eqn:approximatep} represents the third source of error. Our main tool for bounding the sums above will be martingale concentration inequalities.

\subsection{Technical parameters}

For easy reference we collect the parameters that are used in the analysis.

\begin{definition2}[Technical parameters]~
\label{def:constants}
    \begin{itemize}
    \item $\epsilon \in (0,1)$ is the parameter appearing in the statements of~\Cref{thm:randomordermain} and~\Cref{thm:main}.  (Recall that for $\epsilon \geq 1$ the theorems are trivial, so we may assume $\epsilon<1$.)
\item $n$ always represents the number of vertices in the graph.
\item $m = \lfloor n \Delta /2 \rfloor$ is the total number of time steps.
\end{itemize}
In the dense case, there is a \emph{density parameter} $M>1$ that is an upper bound on $n/\Delta$. For
notational convenience we will say that $M=0$ in the random-order case.

There are several parameters given below that arise in the analysis.  All of the parameters
depend on $\epsilon$ and on $M$.  As indicated above, the value $M=0$ is used to refer to
the random case. 
\begin{itemize}
\item $\zeta=\zeta(\epsilon,M)$ is an \emph{error-control parameter}.  In the random-order case,  $\zeta(\epsilon,0)=e^{-20/\epsilon^2}\frac{\epsilon^3}{10}$.  In the dense case, for $M>1$,
$\zeta(\epsilon,M)=\frac{\epsilon^5}{100M}e^{-(5M/\epsilon^2)^2}$.
        \item $\alpha = \alpha(\epsilon,M)=\zeta\frac{\epsilon^{3}}{5}$.
        In  ~\Cref{section:badevents} we define certain bad events for the run of the algorithm.  These bad events are that certain ``error'' quantities associated with the algorithm grow too large.  The margin of error for these quantities is $\alpha\Delta$.
        \item $b=b(\epsilon,M) = \frac{40}{\alpha \epsilon^2}$.  This is the number of phases in the phase-partition for each vertex.
        \item $C=C(\epsilon,M)  = \frac{2000}{\alpha^4}$.  For each subset $S$ of colors we say that a vertex 
        $v$ is $S$-atypical  (\Cref{def:toomuchdriftvertex}) if (very roughly) at some point in the algorithm the fraction of free colors at $v$ that belong to $S$ differs significantly from, $\frac{|S|}{(1+\epsilon)\Delta}$, the overall fraction of colors that belong to $S$.  One of the bad events is that for some color set $S$, the number of $S$-atypical vertices is at least $C$.
        \end{itemize}
        The final parameter is only relevant for the random case:
        \begin{itemize}
        \item $N=N(\epsilon)=\max(400C(\epsilon,0),50b(\epsilon,0))$. 
        The theorem for the random case requires that $\Delta=\Omega(\log n)$.
        The parameter $N$ is the lower bound on $\Delta/\log(n)$ for which
        the result holds.  
    \end{itemize}
\end{definition2}
\begin{remark2}\label{rem:deltasufflarge}
    Throughout the paper, we will be referring to $\Delta$ ``sufficiently large" in comparison to $\epsilon$. For our purpose, $\Delta \geq \frac{10^{10}}{\alpha^8}$ will suffice. This will become relevant in \Cref{sec:FurtherWork}, when we show that our results extend to certain values of $\epsilon$ that are $o(1)$.
\end{remark2}
\section{Properties of the phase partitions}\label{section:errorbounds}

As described above, the algorithm $\mathcal{A}'$ makes use of the phase-counter sequences $\phi_i(v)$ for each vertex.   We use a different choice of $\phi$  for the dense case
and the random-order case.  The parameters $m$ and $b$ are as in \Cref{def:constants}.

\begin{definition2}[Phase Partition Counters]\label{def:phasepartitionfunc}
    We use the following phase partition counter sequences:
    \begin{enumerate}
        \item \textbf{Phase Partition Counters $\phidense=\set{\phidense_{i}}_{i\in [m]}$ for the dense case}: 
        For each vertex $v$, $\phidense_{0}(v)=0$, and for $i \in [m]$ 
        \begin{align*}
        \phidense_{i}(v)=\left\lceil \frac{\card{\timesteps(v)\cap \set{1,\cdots, i}}\cdot b}{\Delta}\right \rceil
        \end{align*}
        Less formally, the counter is the number of edges so far incident to $v$ times $b/\Delta$ rounded up to the nearest integer.
        \item \textbf{Phase Partition Counters $\phirandom=\set{\phirandom_{i}}_{i\in [m]}$ for the random-order case}: For every vertex $v$, $\phirandom_{i}(v)=\left \lceil\frac{i\cdot b}{m}\right \rceil$ for $i\in [m]$.
    \end{enumerate}
\end{definition2}

Our proof that $\mathcal{A}'$ succeeds relies on two properties of the phase-partition counters,
which we call
\emph{balance} and \emph{controlled error}.  Here we will state these properties
and establish conditions under which $\phirandom$ and $\phidense$ satisfy them.
The properties are determined by the edge-ordered graph $(G,\sigma)$ produced by Builder.

\subsection{The balance property}

\begin{definition2}[Balance]\label{def:balance}
A family $\phi(v)$ of phase counter sequences is said to be \emph{balanced} with respect
to $(G,\sigma)$ provided that
for all $v$ and all phases $r$, every $v$-phase contains at most $2\Delta/b$ edges incident on $v$.
\end{definition2}

The definition of $\phidense$ immediately gives that it is balanced with respect to $(G,\sigma)$ for
any graph $G$ of maximum
 degree $\Delta$ and ordering $\sigma$. On the other hand, $\phirandom$ does not satisfy balance for all $(G,\sigma)$
 but for any fixed graph $G$ it satisfies the property for almost all orderings $\sigma$, as we now show.
We use the following concentration bounds.

\begin{lemma2}\cite[Theorem~2.10]{JLR00}\label{lem:chernoffhyper}
    Let $X$ be a hypergeometric random variable with parameters $m,d$ and $k$, where $\mu=\expect{X}=\frac{kd}{m}$, then, 
    \begin{align*}
        \prob{X\geq \mu+t}\leq \exp\paren{-\frac{t^2}{2(\mu+\frac{t}{3})}}.
    \end{align*}
\end{lemma2}

The following lemma shows that for a random ordering, $\phirandom$ is almost certainly balanced.
\begin{lemma2}[Probability  that $\phirandom$ is unbalanced]\label{lem:goodordering}
     Let $\Delta$ be sufficiently large and suppose $G$ is a graph on $n \leq 2^{\Delta/N}$ vertices, where $N$ is given in \Cref{def:constants}, and $m$ edges, where some of the edges may be null edges. Then the fraction of
     orderings $\sigma$ of edges of $G$ such that $\phirandom$ is not balanced with respect
     to $(G,\sigma)$ is at most $\exp\left(-\frac{\Delta}{20b}\right).$
\end{lemma2}
\begin{proof}
Consider $\sigma$ chosen uniformly at random from all orderings. Let $E^r$ be the (multi)-set of edges 
in phase $r$ (which is the same for all $v$ by the definition of $\phirandom$). We say that $E^r$ is a multi-set because there can be multiple null edges. Let $E^r(v)$ be the set of edges in $E^r$ that are incident to $v$ and suppose $\Delta$ is sufficiently large such that $\Delta \geq \frac{10^5}{\alpha^3} \geq 200 b \log b$. Our goal is to show that for all $v\in V$ and $r\in [b]$, $\card{E^{r}(v)}\leq\frac{2\Delta}{b}$ with high probability. $E^r$  is a uniformly random subset of the edges of size $k_r \in \{\lfloor\frac{m}{b}\rfloor, \lceil\frac{m}{b}\rceil\}$ and therefore $\card{E^{r}(v)}$ is a hypergeometric random variable with expectation $\deg(v) \cdot \frac{k_r}{m} \leq \frac{\Delta}{b}\cdot\left(1+\frac{b}{m}\right) \leq \frac{\Delta}{b}\cdot\frac{3}{2},$ (since $\Delta \geq 2b$). Thus, by \Cref{lem:chernoffhyper}, we have 
\begin{align*}
    \prob{|E^r(v)|\geq \frac{2\Delta}{b}} &\leq \exp\left(-\frac{(\Delta/b)^2/4}{3(\Delta/b)+(\Delta/b)/3}\right)\\
    &\leq \exp\left(-\frac{3\Delta}{40b}\right).
\end{align*}
Taking a union bound over all choices of $v\in V$ and $r\in [b]$ gives us that
the probability that $\phirandom$ is not balanced is at most
    \begin{equation*}
        b \cdot 2^{\frac{\Delta}{N}} \cdot \exp\left(-\frac{3\Delta}{40b}\right) \leq 
        \exp\left(-\frac{3\Delta}{40b}+\frac{\Delta}{50b}+\log(b)\right) \leq \exp\left(-\frac{\Delta}{20b}\right)
    \end{equation*} 
where in the first inequality we use the bound on $N$ given by \Cref{def:constants} and for the second inequality, we use the fact that $\log(b) \leq \frac{\Delta}{200 b}$.
\end{proof}

\subsection{The controlled error property}

The second, and more significant property that we want from our
phase counter sequences is \emph{controlled error}.
As indicated in the overview, a key part of the proof is to establish for each color
set $S$, an upper bound on
$\delta^r(v,S)$ that holds for all but constantly many vertices. 
The upper bound will be expressed in terms of the error function
$\widehat{\epsilon}^{\;r}(v)$, which we now define.

The following notation will be useful. For edge $e$ incident on $v$, $e-v$ denotes the other vertex
of $e$.  Also if $j \in \timesteps(v)$ (so that $v \in e_j$) then $s_j(v)=\phi_j(e_j-v)$, the $(e_j-v)$-phase that contains step $j$. Also it is useful to recall the partial order 
on vertex-phase pairs from \Cref{def:partialorder} : $(u,s)\prec (v,r)$ if the last edge in $u$-phase $s$ comes before the last edge of $v$-phase $r$, or $\last{s}{u}<\last{r}{v}$.
\begin{definition2}[Error Bounds]
Let $\zeta>0$.  Given a phase-partition counter $\phi$ for the 
edge-ordered graph $(G,\sigma)$, 
the error function $\widehat{\epsilon}=\widehat{\epsilon}_{\zeta}$ is defined
on vertex-phase pairs $(v,r)$ inductively as follows:
    \begin{align*}
        &\widehat{\epsilon}^{\;0}(v) = 0\\
        & \widehat{\epsilon}^{\;r}(v) = \zeta + \frac{5}{\Delta \epsilon^2}\sum_{j \in \timesteps^{\leq r}(v)} \widehat{\epsilon}^{\;s_j(v)-1}(e_j-v).
    \end{align*}
\end{definition2} 

The recurrence is well-defined because the value
for $(v,r)$ only depends on the value for vertex-phase pairs $(u,s) \prec (v,r)$.
To see this note that the summation is over timesteps $j \in \timesteps^{\leq r}(v)$.
The $(e_j-v)$-phase of $j$ is $s_j(v)$, so $u$-phase $s_j(v)-1$ ends before $j$ which
is before the end of
the $v$-phase $r$. 

\begin{definition2}\label{def:errorcontrolled}
For $\zeta>0$,
a phase partition counter $\set{\phi_{i}}_{i=1}^m$ has \emph{$\zeta$-controlled error} 
with respect to the ordered graph $(G,\sigma)$ provided that for all vertices $v$ phases $r$:
    \[\widehat{\epsilon}_{\zeta}^{\;r}(v) \leq \frac{\epsilon^3}{10}.\]
\end{definition2}

In what follows we show that for any $(G,\sigma)$, (1) if $n \leq \Delta M$ (dense case), the
phase partition counter $\phidense$ has $\zeta(\epsilon,M)$-controlled error with respect
to $(G,\sigma)$, and (2) if $\Delta \geq N\log n$,
 $\phirandom$ has $\zeta(\epsilon,0)$-controlled error with respect to any $(G,\sigma)$ for which $\phi^R$ is balanced.   We start by formulating a bound on $\widehat{\epsilon}$ by unwinding the recursive
definition.

\begin{definition2}[Valid Paths from $v$]
Let $\phi$ be a phase partition function.
    Let $\mathcal{P}^r(v)$ be the set of paths $(x_{0},x_{1},\cdots,x_{t})=(e_{i_1},e_{i_2},\cdots,e_{i_t})$ such that, $x_0=v$, $i_1 \in  \timesteps^{\leq r}(v)$, and for all $1 \leq k < t$, $\phi_{i_{k+1}}(x_k)< \phi_{i_{k}}(x_k)$. That is, $e_{i_{k+1}}$ arrives in an earlier phase of $x_k$ than edge $e_{i_{k}}$. We also include the empty path of length 0. 
\end{definition2}
\begin{proposition}\label{prop:errorrecurrence} For any $\zeta>0$ and phase-vertex pair $(v,r)$,
    \[\widehat{\epsilon}_{\zeta}^{\;r}(v) \leq \zeta \sum_{P \in \mathcal{P}^r(v)}  \left(\frac{5}{\Delta \epsilon^2}\right)^{l(P)}\]
    where $l(P)$ is the length of the path $P$.
\end{proposition}
\begin{proof}
    The proof is by induction with respect to the vertex-phase partial order. The base case where the phase number is 0 is trivial since $\widehat{\epsilon}^{\;0}(v) = 0$ for all $v$. As noted earlier, for any $j \in \timesteps^{\leq r}(v)$, $(e_j-v,s_j(v)-1) \prec (v,r)$.  Furthermore, for $i_0\in \timesteps^{\leq r}(v)$ and for
   any valid path $(e_{i_1},\cdots, e_{i_t}) \in \mathcal{P}^{s_{i_0}(v)-1}(e_{i_0}-v)$ 
with $\phi_{i_0}(x_1)>\phi_{i_1}(x_1)$, the path $(e_{i_0},e_{i_1},\cdots, e_{i_t}) \in \mathcal{P}^r(v)$.
   Thus, we have
   \begin{align*}
   \widehat{\epsilon}^{\;r}(v) & = \zeta + \frac{5}{\Delta \epsilon^2}\sum_{j \in \timesteps^{\leq r}(v)} \widehat{\epsilon}^{\;s_j(v)-1}(e_j-v)\\
       & \leq \zeta + \zeta \cdot \frac{5}{\Delta \epsilon^2}\sum_{j \in \timesteps^{\leq r}(v)} \sum_{P \in \mathcal{P}^{s_j(v)-1}(e_j-v)} \left(\frac{5}{\Delta \epsilon^2}\right)^{l(P)}\\
       & \leq \zeta \sum_{P \in \mathcal{P}^r(v)} \left(\frac{5}{\Delta \epsilon^2}\right)^{l(P)}
   \end{align*}
   where we include the path of length $0$ starting at $v$.
\end{proof}
We now show that the phase counter functions $\phirandom$ and $\phidense$ are $\zeta$-error controlled under suitable conditions.

\begin{proposition}[Bound on $\widehat{\epsilon}$: Dense Case]\label{prop:boundonepsilondense}
For any $\epsilon>0$ and $M >1$, let  $\zeta=\zeta(\epsilon,M)$.  Then
for any graph $G$ with degree $\Delta$ and $n \leq M\Delta$
and any edge ordering $\sigma$ the phase partition function $\phidense$ has $\zeta(\epsilon,M)$-controlled
error.  
\end{proposition}
\begin{proof}
    We first show that there are at most $\frac{n^{t}}{(\lfloor t/2 \rfloor)!}$ valid paths of length $t$ starting from $v$. Note that it is enough to show this for $t$ even since if there are at most $r$ valid paths of length $t$, then there are at most $rn$ valid paths of length $t+1$. Let $P = (e_{i_1}, ..., e_{i_t})$ be any valid path where $v \in e_{i_1}$.  Note that the edges in even positions $e_{i_2}, e_{i_4},...,e_t$ determine the path. Furthermore, these edges can only be placed in reverse arrival order for the path to be valid, so choosing the set of edges determines the path. Therefore, there can be at most $\binom{n^2}{\lfloor t/2 \rfloor} \leq \frac{n^t}{\lfloor t/2 \rfloor!}$ such paths. Thus, if $n \leq M\Delta$, then, by \Cref{prop:errorrecurrence},
    \begin{align*}
        \widehat{\epsilon}^{\;r}(v) & \leq \zeta \sum_{t = 0}^{\infty} \left(\frac{5}{\Delta \epsilon^2}\right)^t \left(\frac{n^t}{(\lfloor t/2 \rfloor)!}\right) \\
        & \leq \zeta \sum_{t' = 0}^{\infty} \frac{2(5M/\epsilon^2)^{2t'+1}}{t'!} \\
        & \leq \frac{10M}{\epsilon^2}\cdot \zeta\sum_{t' = 0}^{\infty} \frac{((5M/\epsilon^2)^2)^{t'}}{t'!} \\
        & \leq  \frac{10M}{\epsilon^2}\cdot \zeta  e^{(5M/\epsilon^2)^2} \leq \frac{\epsilon^3}{10} 
    \end{align*}
    by our choice of $\zeta=\zeta(\epsilon,M)$ from \Cref{def:constants}.
\end{proof}
We remark that there is a trivial bound of $\binom{m}{t}$ on the number of paths.  The
more careful bound of $\frac{n^t}{\lfloor t/2 \rfloor!}$ is crucial in the above analysis.

\begin{proposition}[Bound on $\widehat{\epsilon}$: Random Order Setting]\label{prop:boundonepsilonrandom}
Suppose $G$ is a graph on $n$ vertices with maximum degree $\Delta$. Let $\epsilon>0$ and suppose that $\phi^R$ is balanced with respect to $(G,\sigma)$. Then the phase counter function $\phirandom$ 
is $\zeta$-error controlled provided that $\zeta \leq \zeta(\epsilon,0)$. 
\end{proposition}
\begin{proof}
Recall that under $\phirandom$, the phase counters $\phirandom_{i}(v)$ for all $v\in V$ are updated in lockstep. Consequently, for any edge $e_j=(u,v)$, $\phirandom_{j}(u)=\phirandom_j(v)$.
As before, we bound $\widehat{\epsilon}^{\;r}(v)$  by  bounding $\card{\mathcal{P}^r(v)}$ in \Cref{prop:errorrecurrence}. Consider a valid path $P\in \mathcal{P}^r(v)$, where $P=(e_{i_1},\cdots, e_{i_t})$.  Let $(x_0,x_1,\cdots, x_t)$ be the sequence of vertices such that
$e_{i_j}=(x_{j-1},x_j)$. Recall $P$ has the property that $v=x_0$ and $e_{i_1}\in T^{\leq r}(v)$ and for all $1\leq k<t$, we have, $\phirandom_{i_{k}}(x_{i_k})> \phirandom_{i_{k+1}}(x_{i_k})$. In other words, each $e_{i_{k+1}}$ arrives in an earlier phase of $x_k$ than $e_{i_{k}}$.   Under $\phirandom$ the phase-partition for all vertices is the same so there is an associated unique phase, $r_{k}$ associated to $e_{i_k}$ and $r_{1}>r_{2}>\cdots >r_{t}$. So, we count the number of paths by first picking $r_{i}$'s and then fixing the edges themselves. The number of ways of picking $r_{i}$'s is at most $\binom{b}{t}$. Now we show how to inductively choose $e_{i_k}$'s. The number of ways of choosing $e_{i_1}$ after one has fixed $r_{1}$, is $\frac{2\Delta}{b}$ (by the definition of balance
property in \Cref{def:balance}). Having fixed edge $e_{i_j}$ and $r_{j+1}$ the number of ways of picking $e_{i_{j+1}}$ is at most $\frac{2\Delta}{b}$. Thus, we have, by \Cref{prop:errorrecurrence},
    \begin{align*}
        \widehat{\epsilon}^{\;r}(v) & \leq \zeta \sum_{t = 0}^{b} \left(\frac{5}{\Delta \epsilon^2}\right)^t\binom{b}{t}\left(\frac{2\Delta}{b}\right)^t \\ 
        & \leq \zeta \sum_{t = 0}^{b} \binom{b}{t}\paren{\frac{10}{b\epsilon^2}}^t \\
        & = \zeta \left(1+\frac{10}{b\epsilon^2}\right)^b \\
        & \leq \zeta e^{\left(\frac{10}{\epsilon^2}\right)}\\
        & \leq \frac{\epsilon^3}{10}
    \end{align*}
    where the last inequality follows from our choice of $\zeta(\epsilon,0)$ in \Cref{def:constants} and using the fact that $\zeta\leq \zeta(\epsilon,0)$.
\end{proof}

\section{Background on martingales and supermartingales}
\label{sec:martingales}
In this section we review needed definitions and facts about martingales and supermartingales
including Freedman's concentration inequality. We also prove a novel concentration inequality (\Cref{lem:generalconcentration}) that bounds the probability
that several related supermartingales all deviate significantly from their expectation,
which we prove using Freedman's inequality.

As in the previous section, we consider a random process which determines a sequence
$\{\gamestate_i\}$ where $\gamestate_i$ represents the history of the process through time $i$.
Suppose $\{Y_i\}$ is a sequence of random variables such that $Y_i$ is determined by
$\mathcal{\gamestate}_i$.   
\begin{itemize}
    \item $\{Y_i\}$ is a \emph{martingale} with respect to $\{\mathcal{\gamestate}_i\}$ provided that
    \[\Exp[Y_i \mid \mathcal{\gamestate}_{i-1}] = Y_{i-1}.\]
    \item $\{Y_i\}$ is a \emph{supermartingale} with respect to $\{\mathcal{\gamestate}_i\}$ provided that
    \[\Exp[Y_i \mid \mathcal{\gamestate}_{i-1}] \leq  Y_{i-1}.\]
\end{itemize}

Suppose $\{D_i\}$ is also a sequence of random variables where $D_i$
is determined by $\mathcal{\gamestate}_i$.
\begin{itemize}
    \item $\{D_i\}$ is a \emph{martingale difference sequence} with respect to $\{\mathcal{\gamestate}_i\}$ provided that
    \[\Exp[D_i \mid \mathcal{\gamestate}_{i-1}] = 0.\]
    \item $\{D_i\}$ is a \emph{supermartingale difference sequence} with respect to $\{\mathcal{\gamestate}_i\}$ provided that
    \[\Exp[D_i \mid \mathcal{\gamestate}_{i-1}] \leq  0.\]
\end{itemize}

For any martingale (resp. supermartingale) difference sequence $\{D_i\}$, the sequence
of sums $\{Y_i = \sum_{j = 1}^{i} D_j\}$ form a martingale (resp. supermartingale). Conversely,
for any martingale (resp. supermartingale) $\{Y_i\}$, the differences $\{D_i = Y_i-Y_{i-1}\}$ form a martingale (resp. supermartingale) difference sequence.

We make the following observation about properties of difference sequences that will be useful later on:
\begin{observation2}\label{obs:diffsubsequence}
    Let $\{D_i\}$ be a martingale difference sequence with respect to $\{\mathcal{\gamestate}_i\}$.
    \begin{enumerate}
        \item For any $\ell < i$, conditioning on $\mathcal{S}_{i-1}$ fixes $D_{\ell}$, so we have
    \[\Exp[D_{\ell}D_i \mid \mathcal{\gamestate}_{i-1}] = D_{\ell}(\mathcal{\gamestate}_{i-1})\Exp[D_i \mid \mathcal{\gamestate}_{i-1}] = 0.\]
        \item More generally, let $\{\beta_i\}$ be a sequence of random variables where $\beta_i$ is determined by $\mathcal{\gamestate}_{i-1}$. Then the 
    sequence $\{\beta_iD_{i}\}$ is also a martingale difference sequence with respect to $\{\mathcal{\gamestate}_i\}$:
        \[\Exp[\beta_iD_i \mid \mathcal{\gamestate}_{i-1}] = \beta_i(\mathcal{\gamestate}_{i-1})\Exp[D_i \mid \mathcal{\gamestate}_{i-1}] = 0.\]
    \end{enumerate}
\end{observation2}

We say that $\{\beta_iD_i\}$ is \emph{derived from} the martingale $\{D_i\}$
and refer to $\beta_i$ as the \emph{coefficient sequence} for the derived martingale.
We emphasize that the coefficients here are themselves  random variables.

In this paper, the difference sequences that we consider are indexed by  the time
steps of the process (corresponding to the edges of the graph) and we will 
associate a $\beta$ sequence to each vertex, where
the $\beta$-sequence associated to $v$ is nonzero only on the edges that touch $v$. For our results, we will need to show that certain derived supermartingales associated with particular vertices are unlikely to exceed their expected value by too much. The well-known Azuma-Hoeffding bound provides bounds of this type where the bound (on
the probability of too large a value)
degrades inversely with $\sum_i \lambda_i^2$ where $\lambda_i$ is a bound on the
$i$th difference $D_i$. Unfortunately, in our applications the bounds on the individual $D_i$ are
not sufficient to get a good bound from Azuma-Hoeffding, which does not allow us to
take advantage of the situation we have where most of the $D_i$ are 0, but we don't know
which ones in advance.   Instead we use a variant of Azuma-Hoeffding due to Freedman that is well-suited
for analyzing processes whose evolution is partially controlled by an adversary (which for us is Builder).  In Freedman's theorem one considers the auxiliary sequence $V_i=\text{Var}(D_i \mid \mathcal{\gamestate}_{i-1})$, which is the variance of $D_i$ conditioned on $\gamestate_{i-1}$.
Note that $V_i$ is itself a random variable and Freedman's theorem applies in situations
when $\sum_i V_i$ can be bounded above with probability 1.  (Actually, Freedman's theorem
applies even if the bound on $\sum_i V_i$ fails with small probability, but we only need
a simplified version where $\sum_i V_i$ is always bounded.)

\begin{lemma2}\cite[Theorem~4.1]{Freedman}\label{lem:freedman}
Suppose $\{Y_i\}$ is a supermartingale with respect to $\{\mathcal{\gamestate}_i\}$ and its corresponding difference sequence $\{D_i\}$ satisfies $\card{D_i} \leq D$ for all $i$. Let $V_i = \text{Var}(D_i \mid \mathcal{\gamestate}_{i-1})$ and $W_i = \sum_{j \leq i} V_j$, and suppose $W_{\final} \leq b$ with probability 1. Then

\[\prob{Y_m \geq Y_0+\delta} \leq \exp\left(-\frac{\delta^2}{2(D\cdot\delta/3+b)}\right).\]
\end{lemma2}

In the next section we will use \Cref{lem:freedman} to bound the probability that the derived supermartingale associated to a vertex gets too large. This bound is stated in the first part of the Lemma below. Additionally, there will be times we want to show that for a set of $C$ vertices, $v_1,...,v_C$, the derived supermartingales associated to each of those vertices cannot {\em all} become too large at once. In particular, we would like to show that the probability of this occurring decays exponentially in $C$. Note that if we were guaranteed that the edges adjacent to each of the $v_k$ were disjoint and arrived contiguously - that is, if we could partition the difference sequence $\{D_i\}_{i = 1}^{m}$ into $C$ sequences $\{D_i\}_{i = 1}^{m_1},...,\{D_i\}_{i = m_{C-1}+1}^{m_C}$ such that the derived supermartingale associated to $v_k$ was nonzero only on $\{D_i\}_{i = m_{k-1}+1}^{m_k}$ - then we could get this result by iteratively applying \Cref{lem:freedman}, since conditioned on the value of $\{D_i\}_{i = 1}^{m_{k-1}}$, the sequence $\{D_i\}_{i = m_{k-1}+1}^{m_k}$ is still a difference sequence. However, in our case, the arrival times of the edges for the different vertices can be interleaved, and even intersect. Nevertheless, we manage to provide a general sufficient condition for a similar conclusion to hold.

\begin{lemma2}\label{lem:generalconcentration}
Suppose $\{D_i\}_{i = 1}^{t}$ is a martingale difference sequence with respect to $\{\mathcal{\gamestate}_i\}$ such that $|D_i| \leq 1$ for all $i$ and let $\alpha,a,\Delta$ be
positive real numbers where $\Delta$ is sufficiently large (it will be enough to take $\Delta \geq \frac{2a}{\alpha^2}$). 
\begin{enumerate}
    \item Suppose $\{\beta_i\}_{i=1}^t$ is a coefficient sequence where $\beta_i$ depends only on $\gamestate_{i-1}$ and
such that with probability 1, $\sum_{i=1}^t |\beta_i| \leq \Delta$ and $|\beta_i| \leq a$ for all $i$.
Then:
\[\prob{\sum_i \beta_iD_i \geq \alpha\Delta} \leq \exp\left(-\frac{\alpha^4\Delta}{128a}\right).\]
\item Suppose that for each $k\in \{1,\ldots,C\}$, 
$\{\beta^k_i\}_{i=1}^t$ is a coefficient sequence where $\beta^k_i$ depends only on $\gamestate_{i-1}$ and
such that with probability 1, $\sum_{i=1}^t |\beta^k_i| \leq \Delta$ for all $k$. Suppose further that with probability 1,
for all $i \in \{1,\ldots,t\}$, $\sum_{k=1}^C|\beta^k_i| \leq a$.  Then:
\[\prob{\forall k, \sum_i \beta^k_iD_i \geq \alpha\Delta} \leq \exp\left(-\frac{C\alpha^4\Delta}{128a}\right).\]
\end{enumerate}
\end{lemma2}

The key thing to note about the conclusion is that for $\alpha$, $\Delta$ and $a$ fixed, the probability upper bound shrinks exponentially with $C$.  The first part of the Lemma
is just the case $C=1$ of the second part; we stated it separately to help the reader
to digest the lemma statement, and also because the special case $C=1$ will be applied twice in what follows.

\begin{proof}[Proof of \Cref{lem:generalconcentration}]
The proof is obtained by applying \Cref{lem:freedman} to a single random sequence $\{Y_j\}$ that is constructed from $\{D_i\}$ and all $C$ coefficient sequences.
Let $Y_0=0$, and $j \in \{1,\ldots,t\}$ let:
    \[Y_j = \sum\limits_{k = 1}^{C} \left[\left(\sum_{i \leq j} \beta^k_i D_i \right)^2  - \sum_{i \leq j}(\beta^k_i)^2\right].\]
    If it is the case that for all $k$, $|\sum_{i=1}^t \beta^k_iD_i| \geq \alpha\Delta$ then:
 \[Y_t = \sum\limits_{k = 1}^{C} \left[ \left(\sum_{i \leq t} \beta^k_i D_i \right)^2 - \sum_{i \leq t}(\beta^k_i)^2 \right] \geq C\alpha^2\Delta^2-C\Delta a \geq \frac{C\alpha^2\Delta^2}{2},\]
    since $\Delta \geq \frac{2a}{\alpha^2}$, and therefore:
    \[\prob{\forall k, \sum_i \beta^k_iD_i \geq \alpha\Delta} \leq \prob{Y_t \geq  \frac{C\alpha^2\Delta^2}{2}},
    \]
    so it suffices to bound the probability on the right.  We first show that $\{Y_j\}$
    is a supermartingale.
    Defining $\{Z_j\}$ to be the difference sequence associated to $\{Y_j\}$, we have
    \[Z_j = \sum\limits_{k = 1}^{C} \left[(\beta^k_j)^2D_j^2+2(\beta^k_j)D_j\sum_{i < j} \beta^k_iD_i  - (\beta^k_j)^2\right].\]
    To see that $\{Y_j\}$ is a supermartingale, note that by \Cref{obs:diffsubsequence}, for any $\ell < i$,
    \[\Exp[D_{\ell}D_i \mid \mathcal{\gamestate}_{i-1}] = D_{\ell}\Exp[D_i \mid \mathcal{\gamestate}_{i-1}] = 0.\]
    Furthermore, since $|D_i| \leq 1$ for all $i$, we have $\Exp[D_i^2 \mid \mathcal{\gamestate}_{i-1}] \leq 1$.
    Therefore, for any $1 \leq j \leq t$, we have
    \begin{align*}
        \Exp[Y_j-Y_{j-1} \mid \mathcal{\gamestate}_{j-1}] =\Exp[Z_j \mid \mathcal{\gamestate}_{j-1}] & =  \sum\limits_{k = 1}^{C} \left[(\beta^k_j)^2\Exp[D_j^2 \mid \mathcal{\gamestate}_{j-1}] +2\beta^k_{j}\sum_{i < j} \beta^k_i\Exp[D_jD_i \mid \mathcal{\gamestate}_{j-1}]   -(\beta^k_{j})^2\right]\\
        & =  \sum\limits_{k = 1}^{C}(\beta^k_j)^2\left(\Exp[D_j^2 \mid \mathcal{\gamestate}_{j-1}]-1\right)\\
        & \leq 0,
    \end{align*}
    which shows that $\{Y_j\}$ is indeed a supermartingale. We now
    will use \Cref{lem:freedman} to upper bound the indicated probability.  For this,
    we must bound the variance sums $\{W_j\}$ of $\{Y_j\}$ and the absolute values of the associated difference sequences $\set{Z_i}$. Note:
    \begin{align}
        |Z_j|& \leq \sum\limits_{k = 1}^{C} \left|(\beta^k_j)^2D_j^2+2\beta^k_jD_j\sum_{i < j} \beta^k_iD_i  - (\beta^k_{j})^2\right| \nonumber \\
        & \leq \sum\limits_{k = 1}^{C}\left( 2|\beta^k_j|^2+2|\beta^k_j|\sum_{i < j} |\beta^k_i|\right) & \text{(since $|D_i|\leq 1$)} \nonumber\\
       & = \sum\limits_{k = 1}^{C}2|\beta^k_j| \left(\sum_{i \leq j} |\beta^k_i|\right) \nonumber\\
       &\leq 2\Delta \sum\limits_{k = 1}^{C}|\beta^k_j| 
        & \text{(since $\sum_i |\beta^k_i| \leq \Delta$)} \label{eqn:superimp}\\
        & \leq 2\Delta a & \text{(since $\sum_k |\beta^k_i| \leq a$) \nonumber}
    \end{align}
Thus, 
\[V_j = \text{Var}(Z_j \mid \mathcal{\gamestate}_{j-1}) \leq \Exp[Z_j^2 \mid \mathcal{\gamestate}_{j-1}] \overset{(\ref{eqn:superimp})}{\leq}  2\Delta a \sum\limits_{k = 1}^{C}2\Delta|\beta^k_j|  = 4\Delta^2 a  \sum\limits_{k = 1}^{C}|\beta^k_j|  \]
which tells us
\[W_t \leq \sum_{j} V_j \leq \sum_{j} 4\Delta^2 a  \sum\limits_{k = 1}^{C}|\beta^k_j| = 4\Delta^2 a  \sum\limits_{k = 1}^{C}\sum_{j}|\beta^k_j| \leq 4\Delta^2 a  \sum\limits_{k = 1}^{C} \Delta \leq 4aC\Delta^3 \]
Then \Cref{lem:freedman} tells us that,
    \begin{align*}
        \prob{Y_t \geq  \frac{C\alpha^2\Delta^2}{2}} &\leq \exp\left(-\frac{C^2\alpha^4\Delta^4}{4\cdot(C\alpha^2a\Delta^3+4aC\Delta^3)}\right)\\
        &\leq \exp\left(-\frac{C\alpha^4\Delta}{128a}\right).
    \end{align*}

\end{proof}

\section{Well-Behaved Colorings}\label{section:wellbehavedcolorings}

\subsection{Some Martingales Difference Sequences}\label{section:differencesequences}
Our main goal in this section will be to define the martingale difference sequences we will be considering. Recall that we are viewing the progression of the algorithm as a filtered probability space with $\mathcal{\gamestate}_i$ representing the space of partial colorings of the first $i$ edges to arrive. We first introduce the random variables which will form the basis for the difference sequences we track throughout the course of the algorithm. All of the quantities defined below for an edge $e_i$ will be set to 0 by default if $e_i$ is null or we are unable to color $e_i$.

\begin{definition2}[Collision Variables]
    The following random variables relate to the collisions experienced by algorithm $\mathcal{A}'$.  For a non-null edge $e_i=(u,v)$
    \begin{itemize}
        \item $Z_{i}$ is defined to be 1 if $e_i$ is in a collision  and then successfully colored, and is 0 otherwise. Thus $Z_{i}=1$ provided that there is a collision at $e_i$ (i.e., the preliminary color is not valid) and $F_{i-1}(u) \cap F_{i-1}(v) \neq \emptyset$.
        \item $q_{i}= \Exp[Z_{i} \mid \mathcal{\gamestate}_{i-1}]$.
        If $F_{i-1}(u) \cap F_{i-1}(v) = \emptyset$ this is 0. Otherwise:  
        \[q_{i}: = 1- \frac{|F_{i-1}(u) \cap F_{i-1}(v)|}{|A_{i-1}(u) \cap A_{i-1}(v)|} = \frac{|(A_{i-1}(u) \cap A_{i-1}(v)) \setminus (F_{i-1}(u) \cap F_{i-1}(v))|}{|A_{i-1}(u) \cap A_{i-1}(v)|}.\]
        Since the phase counter functions we use are balanced, at most  $\frac{4\Delta}{b}$ colors from $A_{i-1}(u) \cap A_{i-1}(v)$ could have been used by the time we color $(u,v)$, which gives us:
            \begin{align}\label{eqn:boundcollisions}
                q_{i} \leq \frac{|(A_{i-1}(u)\setminus F_{i-1}(u)) \cup (A_{i-1}(v) \setminus F_{i-1}(v))|}{|A_{i-1}(u) \cap A_{i-1}(v)|} \leq \frac{4\Delta}{b \cdot |A_{i-1}(u) \cap A_{i-1}(v)|}. 
            \end{align}
        \item $\{Z_{i}-q_{i}\}_{i = 1}^{m}$ is a martingale difference sequence with respect to $\{\mathcal{\gamestate}_i\}$, since
            $\Exp[Z_{i}-q_{i} \mid \mathcal{\gamestate}_{i-1}] = 0$. Additionally, since $Z_i,q_i \in [0,1]$, we have for all $i$,
            \[|Z_{i}-q_{i}| \leq 1.\]
    \end{itemize}
\end{definition2}

 The significance of the next set of variables is little more subtle. Recall that our goal is to bound
 \[\delta^r(v,S) := \frac{|A^r(v) \cap S|}{|A^r(v)|} - \frac{|S|}{(1+\epsilon)\Delta},\]
 the error of vertex $v$ with respect to color set $S$ after its $r^{th}$ phase. A natural way to do this would be to track how often the colors chosen for edges incident to  $v$ hit $S$. However, the probability that the color of an edge $(u,v)$ hits $S$ is highly dependent on the palette of $u$, which makes it difficult to control $\delta^r(v,S)$ on its own. Instead we consider  a related, and easier to control, quantity: the difference between the probability of the color of an edge $(u,v)$ hitting $S$ and the indicator for the event. This doesn't directly bound $\delta^r(v,S)$, but it does allow us to approximate $|A^r(v) \cap S|$ in terms of intersections of the form $|A^{r'}(v') \cap S'|$ for neighbors $v'$ of $v$ which - crucially - complete their phase $r'$ before $v$ completes its phase $r$. This will allow us to use an inductive argument to bound $\delta^r(v,S)$ in terms of such $\delta^{r'}(v',S')$. This induction argument is detailed in \Cref{section:mainlemma}.
 
\begin{definition2}[Difference Variables]
Let $e_i = (u,v)$ be a non-null edge and $S \subseteq \Gamma$.
    \begin{itemize}
        \item $X_{i}(S)$ is 1 if the preliminary color for $e_i$ belongs to $S$ and 0 otherwise.
        \item $Y_{i}(S)$ is 1 if the final color chosen for $e_i$ belongs to $S$ and is 0 otherwise.  
        Note that
        \[|X_{i}(S)-Y_{i}(S)| \leq Z_{i},\]
        since the final color chosen for $e_i$ differs from the preliminary color only if there is a collision.
        \item  $p_{i}(S)$ is the probability that the preliminary color chosen for edge $e_i$ is in $S$, conditioned on the coloring of all previous edges:
        \[p_{i}(S) := \Exp[X_{i}(S) \mid \mathcal{\gamestate}_{i-1}] = \frac{|A_{i-1}(u) \cap A_{i-1}(v) \cap S|}{|A_{i-1}(u) \cap A_{i-1}(v)|}.\]
        \item  $D_{i}(S) = X_{i}(S) - p_{i}(S).$
        This is a martingale difference sequence with respect to ${\mathcal{\gamestate}_i}$ since:
        \[\Exp[D_{i}(S) \mid \mathcal{\gamestate}_{i-1}] = \Exp[X_{i}(S) -p_{i}(S) \mid \mathcal{\gamestate}_{i-1}] = 0,\]
        \item Furthermore if $S_1,S_2,\ldots$ is a sequence of color sets
        where $S_i$ is determined by $\mathcal{\gamestate}_{i-1}$ then 
        $\{D_{i}(S_i)\}$ is also a martingale difference sequence, satisfying
        \[|D_{i}(S_i)| = |X_{i}(S_i)-p_i(S_i)| \leq 1\]
        for all $i$. This is important to note down since, often $S_i$ will correspond to palettes of a vertex or edge at a time $i$. These are determined by $\mathcal{S}_{i-1}$.
    \end{itemize}

\end{definition2}

The following proposition relates the variables above to the error $\delta^r(v,S)$ and motivates the difference sequences we will define. For a vertex $v$ and phase $\ell$, let 
\[\widetilde{\timesteps}^{\ell}(v) = \{j : j \in \timesteps^{\ell}(v) \text{ and $e_j$ is successfully colored}\}.\]
Note that $|\widetilde{\timesteps}^{\ell}(v)| = |U^{\ell}(v)|$ and recall that if $e_i$ was not colored, then by definition $Z_i = Y_i(S) = X_i(S) = D_i(S) = p_i(S) = 0$ for all $S$.

\begin{proposition}\label{prop:deltatomartingales}
    For any vertex $v$, subset of colors $S$, and phase $r$ of $v$, $|\delta^r(v,S)|$ is upper bounded by
     \[ \frac{1}{|A^{r}(v)|} \sum_{i \in \timesteps^{\leq r}(v)}Z_{i} ~ + ~ \left|\sum_{\ell = 1}^{r} \frac{1}{|A^{\ell}(v)|}\sum_{i \in \timesteps^{\ell}(v)}D_{i}(S) \right| ~ + ~ \sum_{\ell = 1}^{r} \frac{1}{|A^{\ell}(v)|} \sum_{i \in \widetilde{\timesteps}^{\ell}(v)} \left|p_{i}(S) - \frac{|A^{\ell-1}(v) \cap S|}{|A^{\ell-1}(v)|} \right|.\]
\end{proposition}
\begin{proof}
    From \Cref{eqn:usedcolorssize},
    \[\delta^{\ell}(v,S)-\delta^{\ell+1}(v,S) = \frac{1}{|A^{\ell+1}(v)|}\left(|U^{\ell+1}(v) \cap S| - |A^{\ell}(v) \cap S|\cdot \frac{|U^{\ell+1}(v)|}{|A^{\ell}(v)|}\right).\]
    By definition $\delta^0(v,S) = 0$ and so:
    \begin{align*}
|\delta^r(v,S)| &= \underbrace{\left| \sum_{\ell = 1}^{r} \delta^{\ell-1}(v,S)-\delta^{\ell}(v,S)  \right| }_{(i)}
\end{align*}
By the discussion above, we have,
\begin{align*}
 (i) &= \left| \sum_{\ell = 1}^{r}  \frac{1}{|A^{\ell}(v)|} \left(|U^{\ell}(v) \cap S| - |U^{\ell}(v)|\cdot \frac{|A^{\ell-1}(v) \cap S|}{|A^{\ell-1}(v)|}\right) \right|\\
 &= \left| \sum_{\ell = 1}^{r}  \frac{1}{|A^{\ell}(v)|} \left(\sum_{i \in \timesteps^{\ell}(v)}Y_{i}(S) -|U^{\ell}(v)|\cdot \frac{|A^{\ell-1}(v) \cap S|}{|A^{\ell-1}(v)|}\right) \right|\\
 &= \left| \sum_{\ell = 1}^{r}  \frac{1}{|A^{\ell}(v)|} \left(\sum_{i \in \timesteps^{\ell}(v)}(Y_{i}(S)-X_{i}(S))+\sum_{i \in \timesteps^{\ell}(v)}X_{i}(S) - |U^{\ell}(v)|\cdot \frac{|A^{\ell-1}(v) \cap S|}{|A^{\ell-1}(v)|}\right) \right|\\
 &\leq \sum_{\ell = 1}^{r}  \frac{1}{|A^{\ell}(v)|} \sum_{i \in \timesteps^{\ell}(v)}\left| Y_{i}(S)-X_{i}(S)\right|+\underbrace{\left| \sum_{\ell = 1}^{r}  \frac{1}{|A^{\ell}(v)|} \left(\sum_{i \in \timesteps^{\ell}(v)}X_{i}(S) - |U^{\ell}(v)|\cdot \frac{|A^{\ell-1}(v) \cap S|}{|A^{\ell-1}(v)|}\right)\right|.}_{(ii)} 
\end{align*}
The term $(ii)$ can be rewritten as follows,
\begin{align*}
    (ii)&=\left| \sum_{\ell = 1}^{r}  \frac{1}{|A^{\ell}(v)|} \left(\sum_{i \in \timesteps^{\ell}(v)}(X_{i}(S)-p_{i}(S))+\sum_{i \in \timesteps^{\ell}(v)}p_{i}(S) -|U^{\ell}(v)|\cdot \frac{|A^{\ell-1}(v) \cap S|}{|A^{\ell-1}(v)|}\right) \right|\\
    &=\left| \sum_{\ell = 1}^{r}  \frac{1}{|A^{\ell}(v)|} \left(\sum_{i \in \timesteps^{\ell}(v)}D_{i}(S)+\sum_{i \in \widetilde{\timesteps}^{\ell}(v)} \left(p_{i}(S) - \frac{|A^{\ell-1}(v) \cap S|}{|A^{\ell-1}(v)|}\right)\right) \right|\\
    &\leq \left| \sum_{\ell = 1}^{r}  \frac{1}{|A^{\ell}(v)|} \sum_{i \in \timesteps^{\ell}(v)}D_{i}(S)\right|+\sum_{\ell = 1}^{r}  \frac{1}{|A^{\ell}(v)|} \sum_{i \in \widetilde{\timesteps}^{\ell}(v)} \left|p_{i}(S) - \frac{|A^{\ell-1}(v) \cap S|}{|A^{\ell-1}(v)|}\right|.
\end{align*}
Thus, since we have $\card{Y_i(S)-X_i(S)}\leq Z_i$, we can conclude,
\begin{align*}
    (i)\leq \frac{1}{|A^{r}(v)|} \sum_{i \in \timesteps^{\leq r}(v)}Z_{i} +\left| \sum_{\ell = 1}^{r}  \frac{1}{|A^{\ell}(v)|} \sum_{i \in \timesteps^{\ell}(v)}D_{i}(S)\right|+\sum_{\ell = 1}^{r}  \frac{1}{|A^{\ell}(v)|} \sum_{i \in \widetilde{\timesteps}^{\ell}(v)} \left|p_{i}(S) - \frac{|A^{\ell-1}(v) \cap S|}{|A^{\ell-1}(v)|}\right|.
\end{align*}
\end{proof}

\subsection{Bad Events}\label{section:badevents}
In this section, we will define certain bad events for the run of the algorithm. These bad events are that some ``error quantities" associated with the algorithm grow too large. These bad events, and their likelihood of occurring will be defined in terms of parameters in \Cref{def:constants}. Note that our parameters vary depending on the random-order or dense case. In particular, in the dense case they depend on $M$. 

We now identify three bad events, each associated with one of the three summands
in \Cref{prop:deltatomartingales}.  If none of them occur, we say that the resulting
coloring is \emph{well-behaved}. In this section we show that the coloring is very likely
to be well-behaved.  In the next section we show that in the two situations (an oblivious
strategy that uses an arbitrary graph and random order, or an adaptively chosen dense graph)
a well-behaved coloring will not have any uncolored edges.

The reader is reminded that various technical parameters are collected in~\Cref{def:constants}.
The key parameter in this section is $\alpha$.

The first type of bad event will occur if there are too many collisions at a particular vertex. This event corresponds directly to the first summand in \Cref{prop:deltatomartingales}.
\begin{definition2}[Too Many Collisions, $\mathcal{W}(v)$]\label{def:eventmanycollisions}
   Given a vertex $v$, the bad event $\mathcal{W}(v)$ occurs if there exists a $j \in \{1,\ldots,m\}$ such that:
   \[\sum_{i \in \timesteps(v), i \leq j} (Z_i-q_i) > \alpha \Delta.\]
\end{definition2}

The second type of bad event  relates to the  second summation in \Cref{prop:deltatomartingales}. 
As mentioned earlier, we can't hope to say that the summation is suitably small for all choices of $S$ and $v$ but it will be enough that for all $S$ it is small for all but constantly many $v$.

\begin{definition2}[Too many $S$-atypical vertices, $\mathcal{B}(S)$]\label{def:toomuchdriftvertex}
For a  vertex $v$ and color set $S$, we say that $v$ is \emph{$S$-atypical} if there is a $v$-phase $1 \leq r \leq b$ such that: 
\[\left| \sum_{\ell = 1}^{r}  \frac{1}{|A^{\ell}(v)|} \sum_{i \in \timesteps^{\ell}(v)}D_{i}(S)\right| > \frac{\alpha}{\epsilon}\] 

Let $B(S)$ be the set of $S$-atypical vertices.  We say that the bad event $\mathcal{B}(S)$ occurs if $|B(S)|\geq C$.  (Here $\alpha$ and $C$ are as given in Definition ~\ref{def:constants}.)
\end{definition2}

The final family of bad events helps track ${|F_{i}(u) \cap F_{i}(v)|}$ for an edge $e= (u,v)$. This will ultimately be used to show that no edge runs out of colors.
\begin{definition2}[Too Much Drift at a pair of vertices, $\mathcal{D}(u,v)$]\label{def:toomuchdriftedge}
    Given a pair of vertices $u,v$, let $S_i = F_{i}(u) \cap F_{i}(v)$ be the set of colors free at both $u$ and $v$ at time $i$. Then, the bad event $\mathcal{D}(u,v)$ occurs if, 
    \[\left|\sum_{\substack{i \in \timesteps(u)\cup\timesteps(v) \\ j_1 \leq i \leq j_2}} D_{i}(S_i) \right| > \alpha \Delta\]
    for any $1 \leq j_1 \leq j_2 \leq m$.
\end{definition2}

Next we will bound the probability of too many bad events occurring to show that the algorithm succeeds with high probability.

\begin{definition2}[Well-behaved Coloring]\label{def:well-behavedcoloring}
We say that a \textbf{coloring is well-behaved} if:
\begin{enumerate}
    \item There are no vertices $v$ such that $\mathcal{W}(v)$ occurs. 
    \label{item:badeventW}
    \item There are no sets $S$ such that $\mathcal{B}(S)$ occurs. 
    \label{item:badeventDSv}
    \item There are no pairs of vertices $u,v$ such that $\mathcal{D}(u,v)$ occurs. 
    \label{item:badeventDe}
\end{enumerate}
\end{definition2}

Next, we state the following lemma, which bounds the probability of a coloring induced by $\mathcal{A}'$ not being well-behaved. We emphasize that this lemma applies even for adaptive adversaries, and to sparse graphs,
provided that $n \leq 2^{\frac{\Delta}{N}}$.

\begin{lemma2}\label{lem:coloringwell-behaved}
If $n \leq 2^{\frac{\Delta}{N}}$ and $\Delta$ is sufficiently large, then with probability at least $1-\exp\paren{-\frac{\alpha^4\Delta}{1000}}$, the events \Cref{def:well-behavedcoloring}\ref{item:badeventW}-\ref{item:badeventDe} do not occur, and consequently, the coloring is well-behaved.
\end{lemma2}
For this lemma, $\Delta \geq \frac{10^{10}}{\alpha^8} \geq \max(N^2,(3bC)^{3/2})$ will be sufficiently large.
\begin{proof}
 We show that the coloring is well-behaved by enumerating over each of the conditions \ref{item:badeventW}-\ref{item:badeventDe}, and bounding the probability they fail.
 \begin{enumerate}
     \item Fix a vertex $v$ and time $1 \leq j \leq m$.  Apply the first part of Lemma~\ref{lem:generalconcentration} with
     \[\beta_{i} = \begin{cases} 1 & i \in \timesteps(v), i \leq j\\ 0 & \text{otherwise } \\ \end{cases}. \]
    
    Since the event that $i \in \timesteps(v)$ depends only on $\mathcal{\gamestate}_{i-1}$, the same holds for $\beta_i$. We have  $|\beta_{i}| \leq 1$ and $\sum |\beta_{i}| \leq |\timesteps(v)| \leq \Delta$. Applying the first part of \Cref{lem:generalconcentration} with $a = 1$
    we obtain:
     
     \[\prob{\sum_{i \in \timesteps(v),i \leq j} (Z_i-q_i) > \alpha \Delta } = \prob{\sum_{i = 1}^{m} \beta_i(Z_i-q_i) > \alpha \Delta } \leq \exp\left(-\frac{\alpha^4\Delta}{128}\right).\]
     
        Taking a union bound over at most $m \leq n\Delta \leq \Delta \cdot 2^{\frac{\Delta}{N}}$ choices for $j$, we see that
        \begin{align*}
        \prob{\mathcal{W}(v) \text{ occurs}} \leq \Delta \cdot 2^{\frac{\Delta}{N}} \cdot \exp\left(-\frac{\alpha^4\Delta}{128}\right)
        \end{align*}
        for any vertex $v$. Then taking a union bound over at most $2^{\frac{\Delta}{N}}$ vertices, we get
        \begin{align*}
            \prob{ \exists v ~ \text{s.t.} ~ \mathcal{W}(v) \text{ occurs}} &\leq \Delta \cdot 2^{\frac{2\Delta}{N}}\cdot\exp\left(-\frac{\alpha^4\Delta}{128}\right)\\
            &\leq \exp\paren{-\frac{\alpha^4\Delta}{128}+\frac{2\Delta}{N}+\ln \Delta}\\
            &\leq \exp\paren{-\frac{\alpha^4\Delta}{500}}.
        \end{align*}
Here we use the fact that $\Delta\geq N^2$ to conclude that $\ln \Delta \leq \frac{\Delta}{N}$ and we use the fact that $N \geq C = \frac{2000}{\alpha^4}$.
        
     \item Fix a set $S$ of colors and set $v_1,...,v_C$ of $C$ vertices. By definition
     if $v_1,...,v_C$ are all $S$-atypical, then for each $k \in \{1,\ldots,C\}$ there is a $v_k$-phase $r_k \in \{1,\ldots, b\}$ such that: 
\[\left| \sum_{\ell = 1}^{r_k}  \frac{1}{|A^{\ell}(v_k)|} \sum_{i \in \timesteps^{\ell}(v_k)}D_{i}(S)\right| > \frac{\alpha}{\epsilon}\] 
We can think of this sum as having the form $\sum_{i\geq 1} \beta^k_i D_i(S)$ where
\[
\beta_i^k = \begin{cases} \frac{1}{|A^{\phi_i(v_k)}(v_k)|} & \text{if $i \in \timesteps^{\leq r_k}(v_k)$}\\
0 & \text{otherwise},
\end{cases}
\]
and then we might hope to apply~\Cref{lem:generalconcentration}.  However, the lemma
requires that $\beta_i^k$ be determined by $\gamestate_{i-1}$ and that is not the case here
because $|A^{\phi_i(v)}(v)|$ depends on the number of  edges incident on $v$ through the end of
$v$-phase $\phi_i(v)$ (and whether they are colored or not) and this is not determined by $\gamestate_{i-1}$.

We address this by constructing a family of \emph{fixed coefficient sequences} which is large enough
that one of them agrees with the above coefficient sequence. Now for each choice of $C$ fixed coefficient sequences (one for each vertex) we will apply ~\Cref{lem:generalconcentration}, and
then take a union bound over all such choices.

We note that all of the above coefficients are of the form $1/|A^{\ell}(v)|$ where
$|A^{\ell}(v)|$ is an integer between $\epsilon \Delta$ and $(1+\epsilon)\Delta$.  
Thus  if $v_1,...,v_C$ are all $S$-atypical, then for each $k \in \{1,\ldots,C\}$ there is a $v_k$-phase $1 \leq r_k \leq b$ and for each $\ell$ between 1 and $b$ there is an integer
$s^{\ell}_k\in [\epsilon\Delta,(1+\epsilon)\Delta]$  such
 that:
\[\left| \sum_{\ell = 1}^{r_k}  \frac{\Delta \epsilon}{s^{\ell}_k} \sum_{i \in \timesteps^{\ell}(v_k)}D_{i}(S)\right| > \alpha \Delta\] 
Consider a fixed choice of $r_k$ and $s^{\ell}_k:1 \leq \ell \leq b, k \in \{1,\ldots, C\}$. 

    For each $k \in \{1,\ldots,C\}$, define the coefficient sequence $\beta^k$ by
\[\beta^k_i = \begin{cases} \frac{\epsilon\Delta}{s^{\ell}_k} & i \in \timesteps^{\ell}(v_k) \text{ with $\ell \leq r_k$} \\ 0 & \text{otherwise } \\ \end{cases}. \]
    
    As before, $\{i \in \timesteps^{\ell}(v_k)\}$ is determined by $\phi_i(v)$, which is determined by $\mathcal{\gamestate}_{i-1}$, so the same holds for $\beta^k_i$. Then, since for all $\ell, k$, $\left|\frac{\epsilon\Delta}{s^{\ell}_k}\right| \leq 1$, for all $i$, $\sum_k |\beta^k_i| \leq |\{v_k : i \in \timesteps(v_k)\}| \leq 2$, and for all $k$, $\sum_i |\beta^k_i| \leq |\timesteps(v_k)| \leq \Delta$, taking $a = 2$ in \Cref{lem:generalconcentration} gives us
    \[\prob{\forall k \in \{1,\ldots, C\} \sum_i\beta^k_iD_i > \alpha\Delta} \leq \exp\left(-\frac{C\alpha^4\Delta}{256}\right).\]

    This time we take a union bound over at most $b^C$ choices of $r_k$ for each vertex and at most $(\Delta+1)^{bC}\leq 2^{bC}\Delta^{bC}$ choices of $\{s^{\ell}_k\}$ to get 
    \[\prob{\text{For all } k, ~ v_k \text{ is $S$-atypical}} \leq b^C \cdot \Delta^{bC} \cdot  2^{bC} \cdot \exp\left(-\frac{C\alpha^4\Delta}{256}\right).\]
    Taking another union bound over at most $\binom{n}{C} \leq 2^{\frac{C\Delta}{N}}$ sets of $C$ vertices and $2^{(1+\epsilon)\Delta}$ sets $S$ gives us
    \begin{align*}
        \prob{\exists S, v_1,...,v_C \text{ s.t. } v_k \text{ is $S$-atypical } \forall k} & \leq 2^{(1+\epsilon)\Delta}\cdot 2^{\frac{C\Delta}{N}+bC}\cdot b^C \cdot \Delta^{bC} \cdot \exp\left(-\frac{C\alpha^4\Delta}{256}\right)\\
        & \leq  \exp\left(-\frac{C\alpha^4\Delta}{128}+2\Delta+\frac{C\Delta}{N}+bC+bC\ln \Delta +C \ln b\right)\\
        &\leq \exp\left(-\Delta \right),
    \end{align*}
    since $bC\ln \Delta+bC+C\ln b\leq 3bC\ln \Delta \leq \Delta$ and from \Cref{def:constants} we have $C = \frac{2000}{\alpha^4}$ and $N \geq 400 C$. 
     \item  Fix $u,v\in V$ and $1 \leq j_1 \leq j_2 \leq m$. 
     Define the sequence $\beta$ by:
    \[\beta_i = \begin{cases} 1 & i \in \timesteps(u)\cup\timesteps(v) , j_1 \leq i \leq j_2\\ 0 & \text{otherwise } \\ \end{cases}. \]
    Then, since the event $\{i \in \timesteps(u)\cup\timesteps(v) \}$ depends only on $\mathcal{\gamestate}_{i-1}$, $|\beta_{i}| \leq 1$, and $\sum |\beta_{i}| \leq |\timesteps(u)\cup\timesteps(v)| \leq 2\Delta$, the first part of  \Cref{lem:generalconcentration} with $a = 1$ gives us:
    \[\prob{\left|\sum_{\substack{i \in \timesteps(u)\cup\timesteps(v) \\ j_1 \leq i \leq j_2}} D_{i}(S_i) \right| > \alpha \Delta} \leq \exp\left(-\frac{(\alpha/2)^4
    \Delta}{128}\right).\]
        Taking a union bound over at most $(n\Delta)^2 \leq n^4 \leq  2^{\frac{4\Delta}{N}}$ choices for $j_1,j_2$ gives us
        \begin{align*}
        \prob{\mathcal{D}(u,v) \text{ occurs}} \leq 2^{\frac{4\Delta}{N}} \cdot \exp\left(-\frac{\alpha^4\Delta}{512}\right).
        \end{align*}
     Taking another union bound over at most $n^2 \leq 2^{\frac{2\Delta}{N}}$ vertex pairs gives us
    \begin{align*}
         \prob{ \exists u,v ~ \text{s.t.} ~ \mathcal{D}(u,v) \text{ occurs}} &\leq 2^{\frac{6\Delta}{N}}\cdot \exp\left(-\frac{\alpha^4\Delta}{512}\right)\\
         & \leq \exp\left(-\frac{\alpha^4\Delta}{512}+\frac{6\Delta}{N}\right)\\
        & \leq \exp\left(-\frac{\alpha^4\Delta}{800}\right),
        \end{align*}
         where in the last line we used $N \geq 400C = \frac{800000}{\alpha^4}$.
 \end{enumerate}
 
 Thus, for $\Delta$ sufficiently large, our total probability of a bad event occurring is at most
    \begin{align*}
         \exp\paren{-\frac{\alpha^4\Delta}{500}}+\exp\left(-\Delta\right)+\exp\left(-\frac{\alpha^4\Delta}{800}\right) &\leq \exp\left(-\frac{\alpha^4\Delta}{1000}\right).
    \end{align*}
 
\end{proof}

\section{Main Lemma}\label{section:mainlemma}

We begin with a brief summary of what we've shown so far and what remains to be shown.
For the dense case with a fixed adaptive adversary, we are given a constant $M$ such that 
$n \leq \frac{\Delta}{M}$. We run the game $\mathcal{A'}$  using the phase-counter $\phidense$, which, by \Cref{prop:boundonepsilondense} is $\zeta(\epsilon,M)$-error controlled. Since $n \leq M\Delta \leq 2^{\Delta/N}$ (for $\Delta$ sufficiently large), \Cref{lem:coloringwell-behaved} implies that the coloring produced by $\mathcal{A}'$ is well-behaved with high probability. 

Similarly, if $n \leq 2^{\frac{\Delta}{N}}$, then by \Cref{lem:goodordering}, with high probability, 
for uniformly chosen ordering $\sigma$, $\phirandom$ is balanced with respect to $(G,\sigma)$. Conditioned on $\phi^R$ being balanced with respect to $(G,\sigma)$, guarantees us that by \Cref{prop:boundonepsilonrandom}, $\phirandom$ is $\zeta(\epsilon,0)$-controlled. Finally, in this case also \Cref{lem:coloringwell-behaved} tells us that the coloring produced by $\mathcal{A}'$ is well-behaved with high probability.

In this final section, we will show that if the coloring produced by $\mathcal{A}'$ is well-behaved and $\phi$ is a phase partition counter with $\zeta$-controlled error, then $\mathcal{A}'$ must have successfully produced a proper coloring of $G$. We do this by inductively showing that vertices are \emph{good} according to the following definition. In this definition, and all following definitions in this section, we will assume that we are given $(G,\sigma)$ with its corresponding phase counter $\phi$ and parameter $\zeta$ as defined above, we will denote $\widehat{\epsilon}_\zeta$ simply as $\widehat{\epsilon}$.
\begin{definition2}[Good Vertices]\label{def:goodvertices}
    A vertex $v$ is \emph{good for $S$ during its $r^{th}$ phase} if \label{eqn:vgoodwrtS}
    \begin{align*}
        |\delta^{r-1}(v,S)| \leq \frac{\widehat{\epsilon}^{\;r-1}(v)\cdot \Delta}{|A^{r-1}(v)|}.
    \end{align*}
\end{definition2}
Note that in this definition we say $v$ is good with respect to $S$ during phase $r$ rather than $r-1$, because the palette for $v$ used during phase $r$ is $A^{r-1}(v)$.

\begin{lemma2}[Main Lemma]\label{lem:mainlemma}
Let $\zeta>0$, let phase partition counter $\set{\phi_i}_{i=1}^m$ be balanced and have $\zeta$-controlled error with respect to the ordered graph $(G,\sigma)$. Further suppose that $\Delta$ is sufficiently large. If the coloring is well-behaved, then for all vertex phase pairs $(v,r)$, for all color sets $S$, if $v$ is an $S$-typical vertex, then $v$ is good for $S$ during its $r^{th}$ phase.
\end{lemma2}

We will prove this lemma by induction on the pairs $(v,r)$ according to the order $\prec$ and by bounding each of the three summands in \Cref{prop:deltatomartingales}. The next two propositions relate
these terms to the error terms.

\begin{proposition}\label{lem:errortermsinduction}
For any set $S\subseteq \Gamma$, and any $i$ with $e_i = (u,v)$, $\phi_i(u) = s$, and $\phi_i(v) = r$, the preliminary color set $A_{i-1}(u)\cap A_{i-1}(v) = A^{s-1}(u)\cap A^{r-1}(v)$, satisfies
        \[\card{A^{s-1}(u)\cap A^{r-1}(v)} \geq \frac{\epsilon^2\Delta}{1+\epsilon} - \card{\delta^{s-1}(u,A^{r-1}(v))}|A^{s-1}(u)|.\]    
\end{proposition}
\begin{proof}
By the definition of $\delta^{s-1}(u)$ in (\ref{def:deltavS}) and the fact that $|A^{r-1}(v)|$ is always at least $\epsilon\Delta$,
    \begin{align*}
        |A^{s-1}(u) \cap A^{r-1}(v)| &=  \frac{|A^{s-1}(u)||A^{r-1}(v)|}{(1+\epsilon)\Delta} + \delta^{s-1}(u,A^{r-1}(v))\cdot |A^{s-1}(u)| \\
        &\geq \frac{\epsilon^2\Delta}{1+\epsilon} - |\delta^{s-1}(u,A^{r-1}(v))| |A^{s-1}(u)|,
    \end{align*}
\end{proof}
Note that this bounds the preliminary colors available to $e_i=(u,v)$ in terms of $\delta^{s-1}(u,S)$. Next we establish bounds on $q_{i}$ and $p_{i}(S)$ if we know that $u$ is good with respect to $A^{r-1}(v)$ and $A^{r-1}(v) \cap S$ during its phase $s$. 

\begin{proposition}\label{prop:propertiesofgoodneighbours} 
For $\zeta>0$,
let phase partition counter $\set{\phi_{i}}_{i=1}^m$ be balanced and have \emph{$\zeta$-controlled error} 
with respect to the ordered graph $(G,\sigma)$. Let $e_i = (u,v)$, $\phi_i(u) = s$, and $\phi_i(v) = r$.
    \begin{enumerate}
\item\label{item:propertiesofgoodneighboursa} If $u$ is good for $A^{r-1}(v)$ during its $s^{th}$ phase then, the number of colors available to edge $e_i$ is not too low, that is, 
        \[|A^{s-1}(u) \cap A^{r-1}(v)| \geq \frac{2\epsilon^2\Delta}{5}\]
        \item\label{item:propertiesofgoodneighboursb} If $u$ is good for $A^{r-1}(v)$ during its $s^{th}$ phase then the probability of a collision is low:
        \[q_i \leq \frac{\alpha}{4}\]
        \item\label{item:propertiesofgoodneighboursc} If for a set $S$, $u$ is good for both $A^{r-1}(v)$ and $A^{r-1}(v)\cap S$ then the probability that the preliminary color chosen for edge $e_i$ hits $S$ is close to what we would expect if the color was chosen randomly from $A^{r-1}(v)$:
        \[\left|p_i(S)-\frac{|A^{r-1}(v) \cap S|}{|A^{r-1}(v)|}\right| \leq \frac{5\widehat{\epsilon}^{\;s-1}(u)}{\epsilon^2}.\]
    \end{enumerate}
\end{proposition}
\begin{proof}
    For the first part:
    \begin{align*}
            |A^{s-1}(u) \cap A^{r-1}(v)| &\overset{(\ref{item:prop34})}{\geq}\frac{\epsilon^2\Delta}{1+\epsilon} - |\delta^{s-1}(u,A^{r-1}(v))||A^{s-1}(u)| \\
            &\overset{(\ref{item:uisgood})}{\geq}\frac{\epsilon^2\Delta}{1+\epsilon} - \widehat{\epsilon}^{\;s-1}(u) \Delta \\
            & \overset{(\ref{item:errorcontrol})}{\geq}\frac{\epsilon^2\Delta}{1+\epsilon} - \frac{\epsilon^3\Delta}{10} \\ 
            & \overset{(\ref{item:eps})}{\geq} \frac{2\epsilon^2\Delta}{5} 
        \end{align*}
        Here, \begin{inparaenum}[$(i)$] 
            \item \label{item:prop34}follows from \Cref{lem:errortermsinduction}, \item \label{item:uisgood} is using the fact that $u$ is good for $A^{r-1}(v)$ during phase $s$, \item\label{item:errorcontrol} is by \Cref{def:errorcontrolled} since $\set{\phi_{i}}_{i=1}^m$ is error-controlled, and \item\label{item:eps} is because $\epsilon<1$. 
        \end{inparaenum}
        
        Using this, (\ref{eqn:boundcollisions}), and the definition of $b$ from  \Cref{def:constants}: 
        \[q_i \leq \frac{4\Delta}{ b \cdot |A^{s-1}(u) \cap A^{r-1}(v)| } \leq \frac{10}{b\epsilon^2} \leq \frac{\alpha}{4},\]   
For the last part, again by (\ref{def:deltavS}), for any color set $T$:
    \begin{align}\label{eqn:err2}
    |A^{s-1}(u) \cap T|=
        \delta^{s-1}(u,T)\cdot |A^{s-1}(u)| + \frac{|A^{s-1}(u)||T|}{(1+\epsilon)\Delta},
    \end{align}
    so 
    \begin{align*}
        p_i(S)-  \frac{|A^{r-1}(v) \cap S|}{|A^{r-1}(v)|}    & =\frac{|A^{s-1}(u) \cap A^{r-1}(v) \cap S|}{|A^{s-1}(u) \cap A^{r-1}(v)|} -  \frac{|A^{r-1}(v) \cap S|}{|A^{r-1}(v)|}\\
        &= \frac{\overbrace{|A^{s-1}(u) \cap A^{r-1}(v) \cap S||A^{r-1}(v)|}^{(a)} - \overbrace{|A^{r-1}(v) \cap S||A^{s-1}(u) \cap A^{r-1}(v)|}^{(b)}}{|A^{s-1}(u) \cap A^{r-1}(v)||A^{r-1}(v)|}
    \end{align*}   
We first expand $(a)$, by letting $T=A^{r-1}(v)\cap S$ in \Cref{eqn:err2},
\begin{align*}
    |A^{s-1}(u) \cap A^{r-1}(v) \cap S||A^{r-1}(v)|=|A^{s-1}(u)|\left(\frac{|A^{r-1}(v) \cap S|}{(1+\epsilon)\Delta}+ \delta^{s-1}(u,A^{r-1}(v) \cap S)\right)|A^{r-1}(v)|
\end{align*}
Similarly, expanding $(b)$ by letting $T=A^{r-1}(v)$ in \Cref{eqn:err2}, we have,
\begin{align*}
    |A^{r-1}(v) \cap S||A^{s-1}(u) \cap A^{r-1}(v)|=|A^{s-1}(u)|\left(\frac{|A^{r-1}(v)|}{(1+\epsilon)\Delta}+ \delta^{s-1}(u,A^{r-1}(v))
    \right)|A^{r-1}(v) \cap S|
\end{align*}
Substituting these terms back, and taking absolute value, we can upper bound $\left|p_i(S)-  \frac{|A^{r-1}(v) \cap S|}{|A^{r-1}(v)|}\right|$ as follows:
    \begin{align*}
            & = \left|\frac{|A^{s-1}(u)|\left[\delta^{s-1}(u,A^{r-1}(v) \cap S)\cdot |A^{r-1}(v)| - |A^{r-1}(v) \cap S|\cdot \delta^{s-1}(u,A^{r-1}(v))\right]}{|A^{s-1}(u) \cap A^{r-1}(v)||A^{r-1}(v)|}\right|\\
            & \leq \frac{|A^{s-1}(u)|\left[|\delta^{s-1}(u,A^{r-1}(v) \cap S)|\cdot |A^{r-1}(v)| + |A^{r-1}(v) \cap S|\cdot |\delta^{s-1}(u,A^{r-1}(v))|\right]}{|A^{s-1}(u) \cap A^{r-1}(v)|\cdot |A^{r-1}(v)|}\\
            & \leq \frac{|A^{s-1}(u)|\left[|\delta^{s-1}(u,A^{r-1}(v) \cap S)| + |\delta^{s-1}(u,A^{r-1}(v))|\right]}{|A^{s-1}(u) \cap A^{r-1}(v)|}\\
            & \leq \frac{2\widehat{\epsilon}^{\;s-1}(u)\Delta}{2\epsilon^2\Delta/5} = \frac{5\widehat{\epsilon}^{\;s-1}(u)}{\epsilon^2},
    \end{align*}
    where the final inequality uses the first part and the assumption that $u$ is good for $A^{r-1}(v)$ and $A^{r-1}(v) \cap S$ during phase $r$. 
\end{proof}

We now turn to the proof of the main lemma.

\begin{proof}[Proof of Main Lemma]
Assume the coloring is well-behaved. We proceed by induction on vertex-phase pairs. First note that since $\delta^0(v,S) = 0$ by definition, we know that for all vertices $v$ and sets $S$, $v$ is good for $S$ during its $1^{st}$ phase. Now, for any pair $(v,r)$ with $1 \leq r < b$, we would like to show that for any set $S$ such that $v$ is $S$-typical, $v$ is good for $S$ during its $(r+1)^{th}$ phase.

Since we proceed to by strong induction on the vertex-phase pair, we fix $(v,r)$. Our induction hypothesis states the following.
Suppose that for any ${(u,s) \prec (v,r)}$ and any set $S'$, if $u$ is $S'$-typical, then $u$ is good for $S'$ during its $(s+1)^{th}$ phase. Consider any $S$ such that  $v$ is $S$-typical. Since we know the coloring is well-behaved, we know that for any phase $\ell \leq r$ of $v$,
$|B(A^{\ell-1}(v))|,|B(A^{\ell-1}(v) \cap S)| \leq C$ (recall \Cref{def:toomuchdriftvertex}). Let \[\timesteps_B^{\ell}(v) = \{i \in \timesteps^{\ell}(v) : e_i-v \in B(A^{\ell-1}(v)) \cup B(A^{\ell-1}(v) \cap S)\},\] so that $|\timesteps_B^{\ell}(v)| \leq 2C$. Then, for any $i \in \timesteps^{\ell}(v) \setminus \timesteps_B^{\ell}(v)$, with $u = e_i-v$ and $\phi_i(u) = s+1$, $\last{s}{u} < i \leq \last{r}{v}$, so $(u,s) \prec (v,r)$. Thus, by the inductive hypothesis, since $u$ is both $A^{\ell-1}(v)$-typical and $(A^{\ell-1}(v) \cap S)$-typical, $u$ must be good for both $A^{\ell-1}(v)$ and $A^{\ell-1}(v) \cap S$ during its $(s+1)^{th}$ phase.

Recall that by \Cref{prop:deltatomartingales}, $|\delta^r(v,S)| $ is upper bound by
\[\frac{1}{|A^{r}(v)|} \sum_{i \in \timesteps^{\leq r}(v)} Z_{i} ~ + ~ \left|\sum_{\ell = 1}^{r} \frac{1}{|A^{\ell}(v)|}\sum_{i \in \timesteps^{\ell}(v)}D_{i}(S) \right| ~ + ~ \sum_{\ell = 1}^{r} \frac{1}{|A^{\ell}(v)|} \sum_{i \in \widetilde{\timesteps}^{\ell}(v)} \left|p_{i}(S) - \frac{|A^{\ell-1}(v) \cap S|}{|A^{\ell-1}(v)|}\right|.\]
     Since the coloring is well-behaved, we know that $\mathcal{W}(v)$ does not occur, and therefore,
     \[\sum_{i \in \timesteps^{\leq r}(v)} Z_{i}-q_{i} \leq \alpha\Delta.\]
     By \Cref{prop:propertiesofgoodneighbours}\ref{item:propertiesofgoodneighboursb}, 
     \begin{align} \label{eqn:boundedcollisions2}
         \sum_{i \in \timesteps^{\leq r}(v)} Z_{i} & \leq \alpha\Delta + \sum_{\ell \leq r} \sum_{i \in \timesteps^{\ell}(v) \setminus \timesteps_{B}^{\ell}(v)} q_{i} + \sum_{\ell \leq r} \sum_{i \in \timesteps_{B}^{\ell}(v)} q_{i} \nonumber \\
         & \leq \alpha\Delta + \sum_{\ell \leq r} \sum_{i \in \timesteps^{\ell}(v) \setminus \timesteps_{B}^{\ell}(v)} \frac{\alpha}{4} + \sum_{\ell \leq r} \sum_{i \in \timesteps_{B}^{\ell}(v)} 1 \nonumber \\
         &\leq \alpha\Delta+\frac{\alpha\Delta}{4}+2bC  \nonumber \\
         &\leq 2\alpha\Delta. 
         \end{align}
     Furthermore, since $v$ is $S$-typical (see \Cref{def:toomuchdriftvertex}), we are guaranteed that, 
     \[\left|\sum_{\ell = 1}^{r} \frac{1}{|A^{\ell}(v)|}\sum_{i \in \timesteps^{\ell}(v)}D_{i}(S) \right| \leq \frac{\alpha}{\epsilon} \leq \frac{\alpha\Delta(1+\epsilon)}{\epsilon\cdot|A^r(v)|}.\]
     Finally, again by \Cref{prop:propertiesofgoodneighbours}\ref{item:propertiesofgoodneighboursc},
     \begin{align*}
         &\sum_{\ell = 1}^{r} \frac{1}{|A^{\ell}(v)|} \sum_{i \in \widetilde{\timesteps}^{\ell}(v)} \left|p_{i}(S) - \frac{|A^{\ell-1}(v) \cap S|}{|A^{\ell-1}(v)|} \right| \\
         = &\sum_{\ell = 1}^{r} \frac{1}{|A^{\ell}(v)|} \sum_{i \in \widetilde{\timesteps}^{\ell}(v) \setminus \timesteps^{\ell}_B(v)} \left|p_{i}(S) - \frac{|A^{\ell-1}(v) \cap S|}{|A^{\ell-1}(v)|} \right| + \sum_{\ell = 1}^{r} \frac{1}{|A^{\ell}(v)|} \sum_{i \in \widetilde{\timesteps}^{\ell}(v) \cap \timesteps^{\ell}_B(v)} \left|p_{i}(S) - \frac{|A^{\ell-1}(v) \cap S|}{|A^{\ell-1}(v)|} \right|\\
         = &\sum_{\ell = 1}^{r} \frac{1}{|A^{\ell}(v)|} \sum_{\substack{i \in \widetilde{\timesteps}^{\ell}(v) \setminus \timesteps^{\ell}_B(v)\\ s = \phi(e_i-v)-1}} \frac{5\widehat{\epsilon}^{\;s}(e_i-v)}{\epsilon^2}  + \sum_{\ell = 1}^{r} \frac{1}{|A^{\ell}(v)|} \sum_{i \in \widetilde{\timesteps}^{\ell}(v) \cap \timesteps^{\ell}_B(v)} 1\\
         \leq & \frac{1}{|A^r(v)|} \left(2bC + \sum_{\substack{i \in \timesteps^{\leq r}(v)\\ s = \phi(e_i-v)-1}} \frac{5\widehat{\epsilon}^{\;s}(e_i-v)}{\epsilon^2} \right)
     \end{align*}
     Combining this gives us
     \begin{align*}
         |\delta^r(v,S)| &\leq \frac{1}{|A^r(v)|} \left(2\alpha\Delta+ \frac{1+\epsilon}{\epsilon}\alpha\Delta + 2bC + \frac{5}{\epsilon^2} \sum_{\substack{i \in \timesteps^{\leq r}(v)\\ s = \phi(e_i-v)-1}} \widehat{\epsilon}^{\;s}(e_i-v) \right)\\
         &\leq \frac{\Delta}{|A^r(v)|} \left(\zeta + \frac{5}{\epsilon^2\Delta} \sum_{\substack{i \in \timesteps^{\leq r}(v)\\ s = \phi(e_i-v)-1}} \widehat{\epsilon}^{\;s}(e_i-v) \right) = \frac{\widehat{\epsilon}^{\;r}(v)\Delta}{|A^r(v)|},
     \end{align*}
     where the second inequality holds because $2\alpha+\frac{1+\epsilon}{\epsilon}\alpha+\frac{2bC}{\Delta} \leq \frac{5}{\epsilon}\alpha \leq \zeta$, for $\Delta \geq \frac{10^6}{\alpha^7} \geq \frac{2bC}{\alpha}$ (see \Cref{rem:deltasufflarge}).
\end{proof}

Finally, we use the main lemma to show that in a well-behaved coloring, no edge could have been left uncolored. The intuition behind this proof is as follows: for any pair of vertices $u, v$ as the size of $F_i(u) \cap F_i(v)$ decreases, it becomes less likely for any edge of $u$ or $v$ to choose a color from this set. In particular, in a well-behaved coloring, all but constantly many neighbors in each phase $r$ of $u$ will be $A^{r-1}(u)$-typical, so by the Main Lemma and \cref{prop:propertiesofgoodneighbours}, the number of colors available to most edges adjacent to $u$ will be at least $\frac{2\epsilon^2\Delta}{5}$. Thus, once $|F_i(u) \cap F_i(v)|$ falls below this threshold, the probability of hitting a color in the intersection will be proportional to its size. The number of steps needed to use up all of the colors of the intersection now resembles that of the coupon-collector problem, and we can use a similar argument to show that this number must be greater than $2\Delta$ if the coloring is well-behaved. In the proof below, we partition the edges adjacent to $u$ and $v$ into sets $R_j$ which contain the edges colored until the size of the intersection drops by a factor of two. Roughly, we expect the size of each such set to be proportional to $\Delta$, and this will show that the intersection cannot become empty after only $2\Delta$ edges are colored.

\begin{corollary2}\label{cor:noedgeuncolored}
For $\zeta>0$,
let phase partition counter $\phi = \set{\phi_{i}}_{i=1}^m$ have \emph{$\zeta$-controlled error} 
with respect to the ordered graph $(G,\sigma)$. If $\phi$ is balanced with respect to $(G,\sigma)$ and the coloring produced by the $\mathcal{A}'$ is well-behaved, then for every edge $e=(u,v)$, at all times $i\in [m]$, we have $\card{F_{i-1}(u)\cap F_{i-1}(v)}\geq \frac{\epsilon^2\Delta}{2^{h+1}}$, where $h=\lceil 8/\epsilon^2\rceil$. Thus, no edge is left uncolored. 
\end{corollary2}
\begin{proof}
    Suppose for contradiction  that there is an edge $e = (u,v)$ and time $i$ such that $\card{F_{i}(u) \cap F_{i}(v)} < \frac{\epsilon^2\Delta}{2^{h+1}}$. Let $S_j = F_{j-1}(u) \cap F_{j-1}(v)$ and $e_t$ be the edge adjacent to $u$ or $v$ that uses up the last color so that $\card{S_t}\geq\frac{\epsilon^2\Delta}{2^{h+1}}$ and $\card{S_{t+1}}<\frac{\epsilon^2\Delta}{2^{h+1}}$. This gives us $|S_1| \geq |S_2| \geq \cdots \geq |S_{t}| \geq \frac{\epsilon^2\Delta}{2^{h+1}}$ and $|S_{t+1}| < \frac{\epsilon^2\Delta}{2^{h+1}}$.        
    Let $R = \{k \in \timesteps(u)\cup\timesteps(v) : k \leq t\}$ be the set of arrival times of edges adjacent to either $u$ or $v$ and define the partition $\{R_j\}$ of $R$ to be
$R_0 = \{k \in R : |S_k| \geq \frac{\epsilon^2\Delta}{2}\}$ and
    for $h\geq j \geq 1$, 
    \[R_j = \left\{k \in R : \frac{\epsilon^2\Delta}{2^{j+1}} \leq |S_k| < \frac{\epsilon^2\Delta}{2^{j}}\right\}\]

    Note that $\set{R_j}_{j\geq 0}^h$ are disjoint intervals, and therefore, $|\bigcup_j R_j| \leq |\timesteps(u) \cup \timesteps(v)| \leq 2\Delta$.  We will derive
    a contradiction by establishing lower bounds on the size of each $R_j$, and show their sum exceeds $2\Delta$.
    Since $S_1=(1+\epsilon)\Delta$, $|S_i|$ can drop below $\epsilon^2\Delta/2$ only after
    at least $(1+\epsilon-\epsilon^2/2)\Delta$ edges incident on $u$ or $v$ have arrived, we have
    \begin{equation}
        \label{R0}
        |R_0|\geq (1+\epsilon-\epsilon^2/2)\Delta.
    \end{equation}

    \begin{claim2}
    \label{claim:Rj}
        For each $h \geq j \geq 1$, $|R_j| \geq \frac{\epsilon^2 \Delta}{8}-7\alpha \Delta 2^{j-2}$.
    \end{claim2}
    \begin{proof}
    Since the coloring is well-behaved, for each $1 \leq r \leq b$, $|B(A^r(u))|,|B(A^r(v))| \leq C$, so if we let
    \[R^B = \left\{k \in R : k \in \bigcup_{r = 1}^{b} (B(A^r(u)) \cup B(A^r(v)))\right\},\] 
    we have $|R^B| \leq 2bC$. Furthermore, for any $i \in R \setminus R^B$, we know that $e_i = (u,w)$ or $e_i = (v,w)$ for some vertex $w$. Since $w$ is $A^r(u)$-typical and $A^r(v)$-typical for all $1 \leq r \leq b$, by \Cref{lem:mainlemma}, we know that $w$ is good for $A^r(u)$ and $A^r(v)$ for all $1 \leq r \leq b$, during all of its phases. Then, since $0 \leq Z_i, q_i \leq 1$ for all $i$ and $\mathcal{W}(v), \mathcal{W}(u)$ didn't occur (because the coloring is well-behaved), we have by \Cref{prop:propertiesofgoodneighbours}
    \begin{align*}
        \sum_{i \in R} Z_{i} & \leq \sum_{i \in \timesteps(u),i \leq t} Z_{i} + \sum_{i \in \timesteps(v), i \leq t} Z_{i} & \\
        &\leq 2\alpha\Delta + \sum_{i \in \timesteps(u),i \leq t} q_{i} + \sum_{i \in \timesteps(v), i \leq t} q_{i} & \paren{\text{From \Cref{def:well-behavedcoloring}\ref{item:badeventW} and \Cref{lem:coloringwell-behaved}}}\\
        &\leq 2\alpha\Delta+ 2 \sum_{i \in R}q_i & \paren{\text{From definition of }R}\\
        &\leq 2\alpha\Delta+ 2 \sum_{i \in R \setminus R^B}q_i+ 2 \sum_{i \in R^B}q_i & \\
        &\leq 2 \alpha\Delta+4\Delta \cdot \frac{\alpha}{4}+4bC & \text{(From \Cref{prop:propertiesofgoodneighbours}\ref{item:propertiesofgoodneighboursb})}\\
        &\leq  4\alpha\Delta.
    \end{align*}
    Note that since $Z_i \in [0,1]$ for all $i$, this means that for any $R'\subseteq R$, $\sum_{i \in R'} Z_i \leq 4 \alpha\Delta$.
    
    Similarly, since $\mathcal{D}(u,v)$ did not occur, and the sets $R_j$ are intersections of $\timesteps(u)\cup\timesteps(v)$ with intervals we have for all $j$,
    \begin{align*}
        \sum_{i \in R_j} X_{i}(S_i) \leq \left|\sum_{i \in R_j} p_{i}(S_i)\right| + \left|\sum_{i \in R_j} D_{i}(S_i)\right|  &\leq \sum_{i \in R_j} p_{i}(S_i) + 2\alpha\Delta \\
        &\leq \sum_{i \in R_j\setminus R^B} p_{{i}}(S_i)+2\alpha\Delta+2bC\\
        &\leq \sum_{i \in R_j\setminus R^B} p_{{i}}(S_i)+3\alpha\Delta
    \end{align*}
    for all $j$. Finally, for $i \in R_j \setminus R^B$, suppose without loss of generality that $i \in \timesteps^{r}(v)$, with $x = e_i-v$ and $\phi_i(x) = s$. Then by \Cref{prop:propertiesofgoodneighbours}\ref{item:propertiesofgoodneighboursa},
    \[p_{i}(S_i) = \frac{|S_i|}{|A^{r-1}(v) \cap A^{s-1}(x)|} < \frac{\epsilon^2\Delta/2^{j}}{2\epsilon^2\Delta/5} \leq 5 \cdot 2^{-j-1} \leq 2^{-j+2},\]
    so
    \[  \sum_{i \in R_j} X_{i}(S_i) \leq |R_j\setminus R_B|2^{-j+2}+3\alpha\Delta.\]
    On the other hand, since the colors of at least $\frac{\epsilon^2\Delta}{2^{j+1}}$ edges $e_i$ with $i \in R_j$ must hit $S_i$ for $i\leq h$, we have
    \begin{align*}
        \frac{\epsilon^2\Delta}{2^{j+1}} \leq \sum_{i \in R_j} Y_{i}(S_i) \leq \sum_{i \in R_j} X_{i}(S_i)+\sum_{i \in R_j} Z_{i} \leq \sum_{i \in R_j} X_{i}(S_i)+4 \alpha\Delta.
    \end{align*}

   Therefore,
    \[\frac{\epsilon^2\Delta}{2^{j+1}}-4\alpha\Delta \leq \sum_{i \in R_j}X_{i}(S_i)  \leq |R_j\setminus R_B|\cdot 2^{-j+2}+3\alpha\Delta\]
   so
    \[|R_j| \geq |R_j\setminus R_B| \geq\frac{\epsilon^2\Delta}{8}-7\alpha\Delta\cdot 2^{j-2},\]
    to complete the proof of the claim.
    \end{proof}
Summing
    the lower bounds on $R_j$ for $1 \leq j \leq h$ yields:
    \[\left|\bigcup_{j = 1}^{h} R_j\right| \geq  \Delta-7\alpha\Delta\cdot2^h \geq \left(1-\frac{\epsilon^3}{10}\right)\Delta\]
     since $7\alpha 2^{8/\epsilon^2+1} < 14 \alpha e^{8/\epsilon^2} \leq 14 \zeta e^{8/\epsilon^2} \leq \frac{\epsilon^3}{10}$ by our choice of $\zeta$ from \Cref{def:constants}.
     Combining with (\Cref{R0}) and the fact that $\epsilon\leq 1$, yields the desired contradiction:
     \[\left|\bigcup_{j = 0}^{h} R_j\right| > 2\Delta.\] 
\end{proof}

We can now complete the proofs of \Cref{thm:randomordermain} and \Cref{thm:main}.
According to \Cref{cor:noedgeuncolored}, $\mathcal{A}'$ succeeds provided the
partition function is $\zeta$-controlled with respect to  $(G,\sigma)$ and the coloring is well-behaved.

In the dense case of \Cref{thm:main}, \Cref{prop:boundonepsilondense} ensures
that there is a partition function $\phidense$ that is balanced and $\zeta$-error controlled for $(G,\sigma)$. Note that for any $M$, if $\Delta$ is sufficiently large, then $n \leq M\Delta$ implies $n \leq 2^{\frac{\Delta}{N}}$. Thus, in this case \Cref{lem:coloringwell-behaved} implies that the coloring is well-behaved with probability at least $1-2^{-\alpha^4\Delta/1000}$, so this upper bounds the probability that  $\mathcal{A}'$ (and also $\mathcal{A}$) fails to color $(G,\sigma)$. Thus, letting $\gamma=\alpha^4$, we have our proof for \Cref{thm:main}. 

In the random case of \Cref{thm:randomordermain}, from \Cref{lem:goodordering}, the phase counter $\phi^R$ is balanced with respect to $(G,\sigma)$ with probability at least $1-2^{-\Delta/(20b)}$, and combined with
\Cref{prop:boundonepsilonrandom} this ensures that the partition function
is $\zeta$-controlled with respect to $(G,\sigma)$.  As in the dense case the probability that the
coloring is not well behaved is at most $2^{-\alpha^4\Delta/1000}$.  Therefore for all but at most a
$2^{-\frac{\Delta}{20b}}$ fraction of edge orderings, the probability that  $\mathcal{A}$
succeeds against $\obl(G,\sigma)$ is at least $1-
2^{-\alpha^4\Delta/1000}$. Thus, letting $\gamma_1=\frac{1}{20b}$ and $\gamma_2=\frac{\alpha^4}{1000}$, we have our claim. 

\paragraph{Proof of \Cref{cor:deterministic}.}
\Cref{cor:deterministic} now follows from \cite{Ben-DavidBKTW94}. We give a brief sketch of the argument here. As mentioned above, the outcome of $\mathcal{A}$ against an adaptive adversary can be seen as a two player game between a Builder and Colorer. Note that this is a finite two player game with perfect information and no draws, so either Builder or Colorer must have a deterministic winning strategy. Our result above gives a random strategy for Colorer that ensures a win with high probability, which implies that Colorer must be the one with the deterministic winning strategy. This corresponds to a deterministic online coloring algorithm that succeeds on all graphs.

\section{Further Work}
\label{sec:FurtherWork}
\subsection{A Precise Analysis of Competitive Ratio and Degree Bounds}
\label{subsection:improvement}

Our theorem corresponding to the random order case (\Cref{thm:randomordermain}) states that if $\epsilon\in (0,1)$ is a constant, and $\Delta\geq N\log n$ for some $N=N(\epsilon)$, then we can $(1+\epsilon)\Delta$ online edge color $G$ provided its edges arrive in a random-order. In this section, we will show that if $\Delta=\omega(\log n)$, then we can $(1+o(1))\Delta$ color $G$ under the same conditions. This result is formally stated in \Cref{lem:randomordernew}. 

Similarly, our theorem corresponding to the dense case states that if $M>1$ is any constant, and $\epsilon > 0$ is any constant, then for any $G$ with $\Delta>n/M$ that is adaptively chosen, we can $(1+\epsilon)\Delta$ edge color $G$ with high probability. In this section, we show a stronger statement is possible. That is, there exist slowly growing functions $f(n)$ and $g(n)$ such that even with $M\geq f(n)$ and $\epsilon < 1/g(n)$, the above statement is still true. This result is formally stated in \Cref{lem:newdensecase}.

For the first result, we  assume $\Delta=\omega(\log n)$ and take  $\epsilon=\frac{100}{\sqrt{\ln \Delta -\ln \ln n}}$. Recall that our results in previous sections hold for any $\Delta \geq \frac{10^{10}}{\alpha^8}$ and $n \leq 2^{\frac{\Delta}{N}}$. We will show that even in this new regime, these bounds for $\Delta$ and $n$ hold, and therefore our previous results can be applied here as well. 

\begin{lemma2}[Analogue of \Cref{thm:randomordermain}]\label{lem:randomordernew}
    Let $\Delta=\omega(\log n)$ and $\epsilon=\frac{100}{\sqrt{\ln \Delta -\ln \ln n}}$, and $n$ be sufficiently large, then consider the edge coloring game $\Phi(n,\Delta,\epsilon)$. For any $G$ on $n$ vertices with maximum degree $\Delta$, for all but at most $1-o(1)$ fraction of mappings $\sigma$ of $E(G)$ to $[m]$, $\mathcal{A}$ will defeat the oblivious strategy $\obl(G,\sigma)$ (that is produce a proper edge coloring of $G$) with probability at least $1-o(1)$. 
\end{lemma2}

\begin{proof}[Proof Sketch]
By assumption, $\Delta=\omega(\log n)$, and $\epsilon=\frac{100}{\sqrt{\ln \Delta-\ln\ln n}}$. From this, we can conclude the following values of $\alpha$, $b$, and $C$, 
\begin{align*}
    \alpha(\epsilon,0)=\exp\paren{-\frac{20}{\epsilon^2}}\cdot \frac{\epsilon^3}{10}\cdot \frac{\epsilon^3}{5}\geq \exp\paren{-\frac{40}{\epsilon^2}},
\end{align*}
for all $\epsilon<1$. Substituting the value of $\epsilon$ given in the hypothesis, we have, 
\begin{align*}
    \alpha(\epsilon,0)\geq \paren{\frac{\ln n}{\Delta}}^{\frac{1}{250}}. 
\end{align*}
Similarly, we have, 
\begin{align*}
    b=\frac{40}{\alpha\epsilon^2}\leq \frac{40}{\alpha^2}\leq 40\cdot\paren{\frac{\Delta}{\ln n}}^{\frac{1}{125}}
\end{align*}
and
\begin{align*}
    C=\frac{2000}{\alpha^4}\leq 2000\cdot\paren{\frac{\Delta}{\ln n}}^{\frac{1}{60}}
\end{align*}
Finally, we see that for $n$ sufficiently large,
\[\frac{10^{10}}{\alpha^8} \leq 10^{10} \cdot \paren{\frac{\Delta}{\ln n}}^{\frac{8}{250}} \leq \Delta.\]
Using these bounds and \Cref{lem:goodordering}, we can now conclude 
\begin{align*}
    \prob{\phirandom\text{ is not balanced}}\leq \exp\paren{-\frac{\Delta^{124/125}\cdot (\log n)^{1/125}}{800}+\log n}.
\end{align*}
Moreover, note that
\[N = \max(50b, 400 C) \leq 10^6 \cdot \paren{\frac{\Delta}{\ln n}}^{\frac{1}{60}}\]
so we still have
\[\frac{\Delta}{N} \geq \frac{\Delta^{59/60} \cdot \ln n^{1/60}}{10^6} \geq \ln n \]
for $\Delta = \omega(\log n)$ and $n$ sufficiently large.
Thus, we can upper bound the probability that the coloring is not well behaved by \Cref{lem:coloringwell-behaved}, 
\begin{align*}
     \exp\paren{-\frac{\alpha^4\Delta}{1000}}.
\end{align*}
Since $\Delta=\omega(\log n)$, we can assume that with high probability, $\phirandom$ is balanced with respect to $\sigma$ and the coloring induced by $\mathcal{A}'$ is well-behaved. Conditioned on these events, we know that $\phirandom$ is $\zeta$-controlled provided that $\zeta\leq \zeta(\epsilon,0)$ (by \Cref{prop:boundonepsilonrandom}). Thus, the coloring induced by $\mathcal{A}'$, $\phirandom$ and $\sigma$ satisfy the premise of \Cref{cor:noedgeuncolored}. So, we can conclude that no edge remains uncolored.
\end{proof}

We now proceed to the theorem about dense case. In this case, it is possible to improve our result so that both $\epsilon = o(1)$ and $M = \omega(1)$, as long as they satisfy $\frac{\epsilon^2}{M}\geq \frac{30}{\sqrt{\log \Delta-\log \log n}}$.
\begin{lemma2}[Analogue of \Cref{thm:main}]\label{lem:newdensecase}
Let $\Delta=\omega(\log n)$. Let $\epsilon\in (0,1)$ and $M>1$ be such that $\frac{\epsilon^2}{M}\geq \frac{30}{\sqrt{\log \Delta-\log \log n}}$. Then, for $\Delta>n/M$, for any (possibly adaptive) strategy for online coloring game $\Phi(n,\Delta,\epsilon)$, $\mathcal{A}$ wins (produces an edge coloring of the resulting graph) with probability  $1-o(1)$.
\end{lemma2}
\begin{proof}
    The proof is exactly along the lines of the random order case. We first compute the value of $\alpha$ for the dense case.
    \begin{align*}
        \alpha(\epsilon, M)=\frac{\epsilon^5}{100 M}\cdot \exp\paren{-\paren{\frac{5M}{\epsilon^2}}^2}\cdot \frac{\epsilon^3}{5}\geq \exp\paren{-3\cdot \paren{\frac{5M}{\epsilon^2}}^2}\geq \paren{\frac{\log n}{\Delta}}^{1/12}.
    \end{align*}
Similarly, we have, 
\begin{align*}
    b&=\frac{40}{\alpha\epsilon^2}\leq \frac{40}{\alpha^2}\leq 40\cdot \paren{\frac{\Delta}{\log n}}^{1/6}\\
    C&\leq 2000\cdot \paren{\frac{\Delta}{\log n}}^{1/3}\\
    N &\leq 10^6\cdot \paren{\frac{\Delta}{\log n}}^{1/3}.
\end{align*}
As before, this gives us 
\[\frac{10^{10}}{\alpha^8} \leq 10^{10} \cdot \paren{\frac{\Delta}{\log n}}^{2/3}  \leq \Delta\]
and 
\[\frac{\Delta}{N} \geq \frac{\Delta^{2/3}\log n ^{1/3}}{10^6} \geq \log n\]
for $\Delta \geq 10^9\cdot \log n$.
Thus, we have the probability that the coloring induced by $\mathcal{A}'$ is not well behaved with probability at most,
\begin{align*}
  \exp\paren{-\frac{\alpha^4\Delta}{1000}}.
\end{align*}
Following the outline for the random order case, we can see show that no edge is left uncolored in this case as well, and the algorithm succeeds with high probability.
\end{proof}

\subsection{A Simple Static Algorithm for \((1+\epsilon)\Delta\) Edge Coloring}
\label{subsection:StaticAlgo}

In this subsection, we show that our analysis of the randomized greedy algorithm implies a very simple static algorithm for $(1+\epsilon)\Delta$ edge coloring with a runtime of $O(m\cdot 2^{O(1/\epsilon^2)})$. While there are several recent algorithms which achieve $(1+\epsilon)\Delta$ coloring with runtimes that have a better dependence on $\epsilon$ (\cite{Assadi24,BhattacharyaCPS24,ElkinK24}), we nevertheless state our algorithm because it is very easy to implement. Our presentation is borrowed from \cite{BhattacharyaCPS24}, who gave an $O(m/\epsilon)$ implementation of a folklore static algorithm that $(2+\epsilon)$ edge colors any graph. We use the same parameter dependencies as in \Cref{def:constants}. Additionally, we use the following concentration bounds for sums of geometric random variables. 

\begin{lemma2}[Corollary of Problem 2.4 in \cite{DubhashiP09}]\label{lem:concgeometric}
Let $X_1,X_2,\cdots, X_n$ be independent geometric random variables with probability of success $p$. Let $\mu=n/p$, then we have,
\begin{align*}
    \prob{\sum_{i=1}^n X_i\geq 2n/p}\leq \exp\paren{-2n}.
\end{align*}
\end{lemma2}

\begin{lemma2}
    Let $\epsilon>0$ be a constant, and suppose $N=N(\epsilon)$ be sufficiently large. Suppose $n$ is sufficiently large, $\Delta>N\log n$. Then there is a randomized algorithm, which on input a simple graph $G$ with $n$ vertices and maximum degree $\Delta$, outputs a $(1+\epsilon)\Delta$ coloring of $G$ in time $O(m\cdot 2^{O(\nicefrac{1}{\epsilon^2})})$ with high probability.
\end{lemma2}
\begin{proof}[Proof Sketch]
The algorithm that \cite{BhattacharyaCPS24} propose is as follows. We are given a palette $\mathcal{C}$ of $(2+\epsilon)\Delta$ colors. Every vertex $u\in V$ maintains the palette of used colors $\overline{P(u)}$ as a hash table. This is updated as the algorithm proceeds and an edge incident on $u$ is colored. We color edges by considering them one at a time. When an edge $(u,v)$ is considered, it samples colors from $\mathcal{C}$ until it samples a color $c$ that is unused, that is $c\in \mathcal{C}\setminus \overline{P(u)}\cup \overline{P(v)}$. Observe that at any time, for any edge, $\card{C\setminus \overline{P(u)}\cup \overline{P(v)}}\geq \epsilon\Delta$. Thus, in expectation, every edge $(u,v)$ samples $O(1/\epsilon)$ colors from $\mathcal{C}$ till it finds a color that is available. Finally, if $(u,v)$ samples a potential color $c'$, then in $O(1)$ expected time, one can determine if it is used or not by checking hash tables $\overline{P(u)}$ and $\overline{P(v)}$. Thus, the expected runtime of this algorithm is $O(m/\epsilon)$. 

Note that the above algorithm is precisely the randomized greedy algorithm, but with a larger palette. Every edge $(u,v)$ is colored by a color which is sampled uniformly from its set of available colors. We now adapt this to our setting. In what follows, all the parameters we use correspond to the parameters of the random-order online setting in \Cref{def:constants}. We will first give a bound on the expected runtime, and finally show that a $O_{\eps}(m)$ bound can also be attained with high probability. 

For our purpose, we need to show that if
$\card{\mathcal{C}}=(1+\eps)\Delta$, Then, for any edge $(u,v)$, when it is considered, we are choosing its color from a palette which is large. In order to do this, we will simulate the random-order online setting as follows. We introduce an initialization step. In this, every edge $e\in E$ decides uniformly and independently which one of the $b$ phases it will participate in. This takes time $O(m)$. Then we proceed to color the edges according to the implementation of \cite{BhattacharyaCPS24}. However, we consider the edges in the order of 1) their phase, and 2) within each phase, in the increasing order of their indices. We call this ordering $\sigma$. Observe that $\sigma$ together with choice of phases, defines phase partition counters $\set{\phi_i}_{i\in [m]}$. Our goal is to show that every edge is choosing a color from a palette of size $2^{-O(1/\epsilon^2)}\Delta$. Note that similar to \Cref{lem:goodordering}, one can prove that $\set{\phi_i}_{i\in [m]}$ balanced with respect to $\sigma$. This implies by \Cref{prop:boundonepsilonrandom}, that $\phi$ is $\zeta$-error controlled. Since the graph $G$ satisfies the premise of \Cref{lem:coloringwell-behaved}, the coloring induced by the algorithm is well-behaved. Consequently, the premise of \Cref{lem:mainlemma} is satisfied and therefore it holds and we can argue that \Cref{cor:noedgeuncolored} holds. Thus, we have for every edge $(u,v)$, $\card{\mathcal{C}\setminus \overline{P(u)}\cup \overline{P(v)}}\geq 2^{-O(1/\epsilon^2)}\Delta$ when $(u,v)$ is considered. Consequently, we can now argue that the randomized greedy algorithm with a palette $\mathcal{C}$ such that $\card{\mathcal{C}}=(1+\epsilon)\Delta$ is able to properly edge color a graph with high probability, and with an expected runtime of $O(m\cdot 2^{O(1/\epsilon^2)})$. This argument is exactly as outlined above.

To give a high probability bound, observe that the 
runtime of the algorithm is dominated by the total number of ``attempted" samplings by each edge.
The sampling process of every edge $e_i=(u,v)$ until it hits a color in $\mathcal{C}\setminus \overline{P(u)}\cup \overline{P(v)}$ can be modelled as a geometric random variable $X_i$ with mean $2^{-8/\epsilon^2}$ (by \Cref{cor:noedgeuncolored}). That is, $X_i$ is the number of colors sampled from $\mathcal{C}$ until a color from $\mathcal{C}\setminus \overline{P(u)}\cup \overline{P(v)}$ is sampled. Thus, we have,
\begin{align*}
    \expect{\sum_{i=1}^m X_i}=m\cdot 2^{8/\epsilon^2}.
\end{align*}
Thus, we have, by \Cref{lem:concgeometric},
\begin{align*}
    \prob{\sum_{i=1}^m X_i\geq 2\cdot m\cdot 2^{8/\eps^2}}\leq \exp\paren{-2\cdot m}\leq \exp(-2\Delta)\leq \frac{1}{n^2}.
\end{align*}
This shows our claim.
\end{proof}

\subsection{Generalizing the Dense Case}
\label{subsection:densecase}
Here we take a moment to address the restriction that $n = O(\Delta)$ in \Cref{thm:main}. Although this restriction may seem severe, it appears to be more an artefact of our proof than a natural barrier to the result. We briefly highlight a few specific instances of this below. Note that \Cref{lem:coloringwell-behaved} only requires $n \leq 2^{\frac{\Delta}{N}}$, and in fact the only place density is used is in \Cref{prop:boundonepsilondense}. Since no other part of the proof requires the density restriction, improving the error bound would directly improve the result.
Recall that \Cref{prop:deltatomartingales} upper bounded $|\delta^r(v,S)|$ by
\[ \frac{1}{|A^{r}(v)|} \sum_{i \in \timesteps^{\leq r}(v)}Z_{i} ~ + ~ \left|\sum_{\ell = 1}^{r} \frac{1}{|A^{\ell}(v)|}\sum_{i \in \timesteps^{\ell}(v)}D_{i}(S) \right| ~ + ~ \sum_{\ell = 1}^{r} \frac{1}{|A^{\ell}(v)|} \sum_{i \in \widetilde{\timesteps}^{\ell}(v)} \left|p_{i}(S) - \frac{|A^{\ell-1}(v) \cap S|}{|A^{\ell-1}(v)|} \right|.\]
We can observe the following areas for improvement in this bound:
\begin{itemize}
    \item The error arising from the collisions in the first summation is not strictly necessary in the dense case. By updating the palette after every step and directly considering sets of $C$ neighbors of $v$ and all sets of possible sets of free sets $\{F_i(v)\}$ corresponding to those neighbors, we can eliminate this uncontrolled  source of error.
    \item It would suffice to show the bounds under the assumption that no edge has been left uncolored so far, since once an edge has been skipped by the algorithm, we have already failed. This assumption requires \Cref{cor:noedgeuncolored} to be proven by induction in parallel with the main lemma, but allows us to fix both the size of the free set after any given step and the set of vertices being summed over.
\end{itemize}

When we eliminate the errors from collisions and fix the free set sizes, we are left with an expression that is essentially a weighted sum of difference variables and neighbor error terms. We point out one final, key inefficiency in this accounting of our error:
\begin{itemize}
    \item Applying the triangle inequality to the third summation appears to needlessly inflate the total error. This models the scenario in which the error of every neighbor of $v$ with respect to $S$ is in the same direction, which seems unlikely to be the case in reality.
\end{itemize}

 This gives us hope that we might be able to leverage the observation above to improve the bound on $\delta(v,S)$.

\section*{Acknowledgements} 
We would like to thank Jeff Kahn for helpful discussions.
\bibliographystyle{alpha}
\bibliography{references}
\end{document}